\documentclass[twoside,11pt]{article}

\usepackage[abbrvbib, preprint]{jmlr2e}

\usepackage{amsmath,amsfonts,dsfont}
\usepackage{booktabs}
\usepackage{multirow}
\usepackage{lastpage}

\newtheorem{assumption}{}
\newtheorem{assumptionalt}{}  

\renewcommand{\epsilon}{\varepsilon}

\newcommand{\R}{\mathbb{R}}
\newcommand{\N}{\mathbb{N}}
\newcommand{\Hcal}{\mathcal{H}}
\newcommand{\E}{\mathbb{E}}
\newcommand{\pp}{\mathbb{P}}
\newcommand{\mslab}{m_{\text{slab}}}
\DeclareMathOperator{\var} {Var}

\DeclareMathOperator{\tr} {Tr}
\DeclareMathOperator{\bff} {BF}
\DeclareMathOperator{\gp} {GP}
\DeclareMathOperator{\diag} {diag}
\DeclareMathOperator{\kl} {D_{KL}}

\newcommand\dotprod[2]{\left\langle#1,#2\right\rangle}
\newcommand{\textemph}[1]{\textbf{#1}}

\jmlrheading{}{}{}{}{}{}{José R. Berrendero, Antonio Coín and Antonio Cuevas}

\ShortHeadings{A Bayesian Test for the Mean of Functional Data}{Berrendero, Coín and Cuevas}

\begin{document}

\title{A Bayesian One-Sample Test for the Mean of Functional Data}

\author{\name José R. Berrendero \email joser.berrendero@uam.es \vspace{.2em}\\
        \name Antonio Coín \email antonio.coin@protonmail.com \vspace{.2em}\\
        \name Antonio Cuevas \email antonio.cuevas@uam.es \vspace{.2em}\\
       \addr Departamento de Matemáticas\\
       Universidad Autónoma de Madrid\\
       28049 Madrid, Spain
}

\maketitle

\begin{abstract}
In the context of functional data analysis, we propose a Bayesian test to assess whether the mean function of a Gaussian population is identically zero. By leveraging the structure of the reproducing kernel Hilbert space (RKHS) associated with the covariance function of the underlying process, we approximate the mean by a finite linear combination of kernel sections. This finite-dimensional representation arises from the natural assumption that the mean lies in the corresponding RKHS, enabling a parametric Bayesian approach to hypothesis testing through a spike-and-slab prior on the coefficients. Given the absence of a canonical reference measure on infinite-dimensional spaces, the likelihood is expressed via a Radon--Nikodym derivative between induced Gaussian measures. Since the number of kernel sections is unknown, we employ reversible jump Markov chain Monte Carlo (RJMCMC) to explore the trans-dimensional parameter space and estimate the posterior distribution. From the resulting posterior samples we construct estimators of the posterior probability of the null hypothesis and the Bayes factor, providing quantifiable evidence against the null and supporting informed decisions under uncertainty. Furthermore, to validate this methodology we establish near-optimal posterior contraction rates for the proposed test under mild conditions, demonstrate its empirical performance on simulated data sets, and illustrate its application with real climate data.
\end{abstract}


\section{Introduction}

Functional data analysis (FDA) has emerged as one of the most active areas of statistical research over the past decades, driven by the increasing availability of data collected over time, space, or other continuous domains. In contrast to classical multivariate statistics, where each data point is a vector, the FDA framework treats each observation as a curve, surface, or more general function, typically viewed as a realization of an underlying stochastic process. This paradigm shift has enabled more nuanced analyses in diverse fields including climatology, biomechanics, neuroimaging, and econometrics, where the fundamental units of observation are inherently functional \citep{ullah2013applications}. The monograph by \citet{ramsay2005functional} is generally regarded as a standard reference for the subject, while detailed up-to-date accounts of the theory and applications are provided for example by \citet{horvath2012inference}, \citet{hsing2015theoretical} or \citet{wang2016functional}.

A core problem in statistics involves hypothesis testing for quantities of interest. As such, considerable efforts have been devoted to generalizing classical testing frameworks to functional data, providing \(p\)-values and decision rules based on asymptotic approximations or bootstrap resampling. These methods include frequentist two-sample tests \citep{fan1998test,hall2007two,zhang2010two}, analysis of variance methods \citep{cuevas2004anova, cuesta2010simple, zhang2013analysis}, change point detection procedures \citep{berkes2009detecting,hormann2010weakly}, repeated measures analysis \citep{martinez2011repeated,smaga2020note} or dependence tests in linear regression \citep{cardot2003testing,kokoszka2008testing}, among others. However, the conceptually simpler one-sample tests, often used to detect the presence or absence of a certain effect within a given population, have received less attention in the FDA literature. Notable exceptions include proposals by \citet{mas2007testing}, \citet{munk2008one} and \citet{zhang2013analysis}, where the authors explore different tests for the mean of a functional sample.

On the other hand, Bayesian methods for hypothesis testing---introduced in modern form by \citet{jeffreys1961theory}---are extensively studied and applied in multivariate settings, offering a principled alternative to the often criticized \(p\)-value approach \citep[see][and references therein]{robert2007bayesian, robert202250}. However, despite their broad adoption, they remain comparatively underdeveloped within the FDA framework. We aim to bridge this gap by proposing a Bayesian test for the mean function \(\mu\) of a Gaussian process observed through i.i.d.\ samples, which serves as a building block for more complex statistical analyses and machine learning applications. In particular, we provide a method for testing the point-null hypothesis \(H_0: \mu \equiv \mu_0\) versus the composite alternative \(H_1: \mu \not\equiv \mu_0\), where \(\mu_0(\cdot)\) is a predetermined function. We focus on the simple hypothesis of exact equality, in contrast to \textit{precise} or \textit{relevant} hypotheses such as \(H_0^\epsilon: \|\mu-\mu_0\| < \epsilon\) for a given \(\epsilon>0\). Although these precise hypothesis may be useful when an exact zero is not meaningful or reasonable, in a Bayesian framework they can be conceptually encoded in the corresponding point-null hypothesis via prior beliefs, as argued by \citet{berger1987testing}. The proposed test is based on the so-called \textit{Bayes factor}, which is defined as
\[
\bff_{10} = \frac{\pp(\text{data}|H_1)}{\pp(\text{data}|H_0)}
\]
and measures the evidence for \(H_1\) over \(H_0\). It is linked to the prior and posterior odds through the formula
\[
\frac{\pp(H_1|\text{data})}{\pp(H_0|\text{data})} = \bff_{10} \, \times\, \frac{\pp(H_1)}{\pp(H_0)}.
\]
Using Bayes factors to quantify the evidence against the null hypothesis in a continuous manner is the standard approach in Bayesian testing, popularized since the seminal work by \citet{kass1995bayes}. Because they are intimately related to the posterior probability of the null, Bayes factors provide a direct probabilistic way of interpreting the test decision, while also allowing the incorporation of prior information. Moreover, despite being a loose analog of a likelihood ratio test, the Bayes factor \(\bff_{10}\) can be inverted to provide evidence in favor of \(H_0\), which is not possible in frequentist tests due to the built-in asymmetry. 

Most functional tests in the literature employ some form of dimension reduction via projections or basis expansions as a first step. We take a slightly different approach and follow earlier work by \citet{berrendero2018use,berrendero2025bayesian}. In particular, we directly assume that the mean function is a member of the reproducing kernel Hilbert space (RKHS) associated with the covariance function of the underlying Gaussian process, so that it can be approximated through a finite linear combination of kernel sections. We will see that this assumption is reasonable and allows us to completely specify a Bayesian test, circumventing the difficulty of placing a prior distribution on an infinite-dimensional space, and providing a closed-form likelihood expression via Radon--Nikodym derivatives. Since we do not restrict ourselves to conjugate models, we use the reversible jump variant of Markov chain Monte Carlo methods (RJMCMC), proposed by \citet{green1995reversible}, to approximate the posterior distribution. From this approximation we obtain estimators of the posterior probability of the null hypothesis and the Bayes factor, which can be used to guide decision-making while explicitly accounting for uncertainty.

To validate the methodology from a theoretical point of view, we study the large-sample consistency of our test and measure the speed of posterior concentration. This technique has become a standard tool in the frequentist assessment of Bayesian procedures, as exemplified by the highly relevant work of \citet{ghosal2017fundamentals}. Posterior consistency of the Bayes factor in a point-null testing situation was established by \citet{dass2004note}, but their abstract approach cannot provide posterior contraction rates. As a novel contribution, we show that under \(H_0\) the posterior probability of the alternative goes to zero at a polynomial rate, while under \(H_1\) the posterior probability of the null decays exponentially fast. On the other hand, finite-sample performance is evaluated by numerical simulations in terms of frequentist metrics, such as type I error (false positive rate), type II error (false negative rate) or power (true positive rate).

\subsection{Preliminaries and Notation}\label{sec:preliminaries}

Consider a fixed time horizon \(T>0\) and suppose our data comes from a Gaussian process \(X=\{X(t): t\in[0,T]\}\), with mean function \(\mu(t)=\E[X(t)]\) and covariance function \(K(t,s)=\E[(X(t)-\mu(t))(X(s)-\mu(s))]\). For convenience, assume that the sample paths of \(X\) are continuous almost surely, so point evaluation is well-defined. We will mainly work with functional observations, but whenever discretization is needed we will suppose that the process is observed on an arbitrarily dense grid \(\mathcal S=\{s_1,\dots, s_G\}\subset [0,T]\) of size \(G\).

To construct the RKHS \(\Hcal(K)\) associated with \(K\), we first consider the functional space \(\Hcal_0(K)\) of all finite linear combinations of evaluations of \(K\), that is,
\[
\Hcal_0(K) = \left\{\sum_{j=1}^p \beta_j K(t_j, \cdot):\, p\in \N,\, \beta_j \in \R,\, t_j\in [0,T]\right\}.
\]
This space is endowed with the inner product \(\dotprod{f}{g}_K = \sum_{j,k}\beta_j \gamma_k K(t_j, u_k)\), given that \(f(\cdot) = \sum_j \beta_j K(t_j, \cdot)\) and \(g(\cdot) = \sum_k \gamma_k K(u_k, \cdot)\). Then, \(\Hcal(K)\) is defined as the completion of \(\Hcal_0(K)\) with respect to this inner product. We remark that functions on this space satisfy the reproducing property: \(\dotprod{K(t, \cdot)}{f}_K = f(t)\) for all \(f\in\Hcal(K)\) and \(t\in[0,T]\); we refer the reader to the monograph by \citet{berlinet2004reproducing} for further details.

As it turns out, this RKHS is isometrically isomorphic to the closed linear span of the process \(X\) in the space of all random variables with finite second moment. This congruence is realized through \textit{Loève's isometry} \(\Psi_X\) \citep[Theorem 35]{berlinet2004reproducing}, which can be seen as the completion of the correspondence \(\sum_j \beta_j K(t_j, \cdot) \longleftrightarrow \sum_j \beta_j X(t_j)\). Despite the close relation between the process and the space \(\Hcal(K)\), concrete trajectories \(x=X(\omega)\) do not belong to \(\Hcal(K)\) with probability one, except in trivial cases \citep[Corollary 7.1]{lukic2001stochastic}. Thus, the expression \(\dotprod{x}{f}_K\) for \(f\in\Hcal(K)\) is ill-defined when \(x\) is a realization of \(X\). However, we can still make sense of this expression by following the approach of \citet{parzen1961approach}, identifying \(\dotprod{X}{f}_K\) with the random variable \(\Psi_X(f)\) through Loève's isometry. We adopt this interpretation hereafter.

Throughout this paper we will work with column vectors, denoting the transpose of a vector \(v\) by \(v'\) and its Euclidean norm by \(\|v\|\). By a slight abuse of notation, arbitrary matrix norms will be equally denoted by \(\|\cdot\|\), while the symbol \(\|\cdot\|_K\) is reserved for the norm of \(\Hcal(K)\). An uppercase \(\Pi\) will stand for prior distributions or their corresponding densities, and similarly \(\Pi_n\) will represent a posterior distribution (or density) based on \(n\) observations. Moreover, by convention, boldface letters will designate discretized functional quantities on the grid \(\mathcal S\).

\subsection{Organization of the Paper} 

In Section~\ref{sec:rkhs-test} we describe the proposed methodology and RKHS modeling assumptions for our Bayesian test. Section~\ref{sec:consistency} explores the asymptotic properties of the test from a frequentist perspective, and Section~\ref{sec:experiments} is concerned with numerical experiments and a real data application. In Section~\ref{sec:conclusions} we summarize the main findings and discuss future work. A description of the proposed RJMCMC scheme is included in Appendix~\ref{app:rjmcmc}, while Appendix~\ref{app:proofs} contains the proofs of the theoretical results. Lastly, further details on the experimentation framework are available in Appendix~\ref{app:experiments}.

\section{A Reproducing Kernel Approach to Bayesian Testing}\label{sec:rkhs-test}

Let \(X_{1:n}=(X_1,\dots, X_n)\) be an i.i.d.\ sample of a Gaussian process \(X\sim \gp(\mu, K)\) with mean \(\mu\) and strictly positive definite covariance function \(K\), which is assumed to be known for the time being. Without loss of generality we set \(\mu_0\equiv 0\) and \(T=1\), which leads us to consider the hypotheses
\[
H_0: \mu(t)=0 \text{ for all } t \in [0,1] \quad \text{against} \quad H_1: \mu(t)\neq 0 \text{ for some } t \in [0,1].
\]

Our key modeling assumption is that \(\mu\in\Hcal(K)\). Not only is this space naturally linked to the process \(X\), but in general it contains smoother functions than, say, \(L^2[0,1]\). More importantly, it provides a direct finite-dimensional representation of the mean function if we further suppose that \(\mu\in\Hcal_0(K)\), which is dense in \(\Hcal(K)\) by construction. Indeed, we propose a parametric model (with a functional interpretation) based on the approximation
\begin{equation}\label{eq:mu-h0}
\mu_\theta(\cdot) = \sum_{j=1}^p \beta_j K(t_j, \cdot),
\end{equation}
where \(\theta=(p, \tau_p, b_p)\) is a generic parameter, with \(\tau_p=(t_1,\dots, t_p)\in [0,1]^p\) the time instants, \(b_p=(\beta_1,\dots,\beta_p)\in \R^p\) the coefficients and \(p\in\N\) the dimensionality (which is not fixed). We will omit the subscript \(p\) when it is understood from context, and assume that the grid of observation \(\mathcal S\) is fine enough to contain all the relevant time instants.

Note that when \(K\) is strictly positive definite, the null hypothesis \(\mu_\theta\equiv 0\) corresponds to \(b=0\) as long as the time instants are distinct. Thus, the parameter spaces for the null and alternative hypotheses in our Bayesian test are
\[
\Theta_0=\bigcup_{p\in\N}\{b \in \R^p: b=0\}\quad \text{and}\quad \Theta_1=\bigcup_{p\in\N}\{(p, \tau, b): \tau \in \mathcal T_p,\, b \in \R^p\setminus \{0\}\},
\]
respectively, where \(\mathcal T_p = \{(t_1,\dots,t_p)\in [0,1]^p: t_1 < \cdots < t_p\}\) is the ordered unit simplex. Although not strictly needed, we impose an ordering constraint to avoid identifiability issues and reduce the parameter space by a factor of \(p!\) on each dimension.

\subsection{Likelihood and Prior Distributions} 

Let \(\pp_\mu\) be the distribution of \(\gp(\mu, K)\) on the path space \(\mathcal C[0,1]\) (continuous functions on \([0,1]\)), and similarly let \(\pp_0\) be the distribution of \(\gp(0, K)\). Since both distributions are Gaussian with common covariance, the Feldman--Hájek dichotomy \citep{feldman1958equivalence} implies that \(\pp_\mu\) is absolutely continuous with respect to \(\pp_0\) whenever \(\mu\in\Hcal(K)\), which is precisely our assumption on the mean function. In this situation, the density of \(\pp_{\mu_\theta}\) with respect to \(\pp_0\) is given by \citep[Theorem 1]{berrendero2018use}
\begin{equation}\label{eq:rn-derivative}
  \frac{d\pp_{\mu_\theta}}{d\pp_0}(X) = \exp\left\{\langle X, \mu_\theta\rangle_K - \frac{1}{2}\|\mu_\theta\|^2_K\right\}.
\end{equation}
In a Bayesian setting, this Radon--Nikodym derivative can be interpreted as the likelihood \(L(X|\theta)\) of observing a path \(X\) under the parameter \(\mu_{\theta}\), with respect to the common reference measure \(\pp_0\). Now, since \(\mu_{\theta}\) has the finite-dimensional representation~\eqref{eq:mu-h0}, we can evaluate the scalar product (via Loève's isometry) and the RKHS norm, obtaining closed-form expressions \(\langle X, \mu_\theta\rangle_K = \sum_{j=1}^p \beta_j X(t_j)\) and \(\|\mu_\theta\|^2_K = \sum_{j=1}^p \sum_{k=1}^p \beta_j\beta_k K(t_j, t_k)\). In particular, for our i.i.d.\ sample the joint likelihood becomes
\begin{equation}\label{eq:likelihood}
  L(X_{1:n}| p, \tau, b) = \exp\left\{S_n(\tau)'b - \frac{n}{2}b' K_{\tau} b\right\},
\end{equation}
where \(S_n(\tau)=\sum_{i=1}^n (X_i(t_1), \dots, X_i(t_p))'\) and \(K_{\tau}=[K(t_j, t_k)]_{j,k=1}^p\). 

Then, for a random variable \(\theta=(p,\tau,b)\) we consider the parametric hierarchical model
\[
  X_1,\dots, X_n | \theta \stackrel{\text{i.i.d.}}{\sim} \gp(\mu_\theta, K),\quad
  p \sim \Pi(p),\quad \tau|p \sim \Pi(\tau|p),\quad b|p \sim \Pi(b|p),
\]
where the implied independence between \(\tau\) and \(b\) (given \(p\)) is just a convenience. We impose a discrete prior distribution \(\Pi(p)\) on the number of components, and a continuous distribution \(\Pi(\tau|p)\) for the time instants (e.g.\ uniform on \(\mathcal T_p\)). The prior on the coefficient vector \(b\) has to be chosen carefully, since any continuous distribution will give zero mass to the null event \(\{b=0\}\). To remedy this, we follow a \textit{spike-and-slab} approach reminiscent of Bayesian variable selection in linear regression \citep[e.g.][]{mitchell1988bayesian,ishwaran2005spike}, whereby a fixed mass is assigned to the null hypothesis and the remaining mass is distributed continuously on the alternative. Let \(0<\pi_0<1\) be the prior belief that the null hypothesis is true, and consider
\begin{equation}\label{eq:prior-b}
\Pi(b|p) = \pi_0\Pi^{\text{spike}}(b) + (1-\pi_0)\Pi^{\text{slab}}(b) = \pi_0 \delta_0(b) + (1-\pi_0)\mathcal N_p(b|0, \eta^2 I_p),
\end{equation}
where \(\delta_0\) is a Dirac measure at \(0\) (the ``spike''), \(I_p\) is the identity matrix of order \(p\), and \(\eta^2>0\) is a scalar hyperparameter. The mixture notation is a shorthand for a hierarchical representation with a latent indicator \(Z\sim \text{Bernoulli}(\pi_0)\) and conditional distributions 
\[
  b| Z=1, p \sim \delta_0\quad \text{and}\quad b| Z=0, p \sim \mathcal N_p(0, \eta^2 I_p).
\]
Unlike the priors for \(p\) and \(\tau\), the choice of a normal distribution for the ``slab'' prior of \(b\) is not arbitrary; it will be used as a conjugacy device to simplify the posterior distribution later on. Note that this construction induces the priors \(\Pi(H_0)=\pi_0\) and \(\Pi(H_1)=1-\pi_0\), since the null hypothesis is reinterpreted as \(H_0: b=0\).

\subsection{Posterior Approximation} 

Bayes' formula gives a posterior distribution \(\Pi_n(p, \tau, b | X_{1:n}) \propto \Pi(p, \tau, b)L(X_{1:n}|p, \tau, b)\) up to a multiplicative constant, which can be approximated through computational methods. However, sampling the high-dimensional parameter \(b\in \R^p\) is notoriously challenging, so we marginalize it analytically thanks to the prior structure in~\eqref{eq:prior-b}. The marginal likelihood has the form
\begin{align*}
  m(X_{1:n}|p, \tau) &= \pi_0\int \Pi^{\text{spike}}(b)L(X_{1:n}|p, \tau, b)\, db + (1-\pi_0)\int \Pi^{\text{slab}}(b)L(X_{1:n}|p, \tau, b)\, db\\
  &= \pi_0m_{\text{spike}}(X_{1:n}|p, \tau) + (1-\pi_0)\mslab(X_{1:n}|p, \tau).
\end{align*}
Since \(\Pi^{\text{spike}}(b)\) is a point mass at \(b=0\), we have \(m_{\text{spike}}(X_{1:n}|p,\tau)=e^0=1\), according to~\eqref{eq:likelihood}. In the slab part we get an exact integral, applying the formula for the moment generating function (at \(t=1\)) of a quadratic form in Gaussian variables \citep[Corollary 3.2a.1]{mathai1992quadratic}:
\begin{equation}\label{eq:m-slab}
  \begin{split}
  \mslab(X_{1:n}|p, \tau) &= \int \mathcal N_p(b|0, \eta^2 I_p)L(X_{1:n}|p, \tau, b)\, db \\
  &= \eta^{-p} |\Lambda_\tau|^{-1/2}\exp\left\{\frac{1}{2}S_n(\tau)' \Lambda_\tau^{-1}S_n(\tau)\right\},
  \end{split}
\end{equation}
where \(\Lambda_\tau = nK_\tau + \eta^{-2}I_p\). All in all, the marginal posterior expression becomes
\begin{equation}\label{eq:posterior-p-tau}
\Pi_n(p, \tau|X_{1:n}) \propto \Pi(p)\Pi(\tau|p)[\pi_0 + (1-\pi_0)\mslab(X_{1:n}|p, \tau)].
\end{equation}

At this point we choose to use the RJMCMC algorithm to approximate the marginal posterior \(\Pi_n(p, \tau|X_{1:n})\). Unlike standard MCMC approaches, this method can traverse a variable-dimension parameter space by proposing the birth of a new component or the death of an existing one at each iteration. This yields a sequence of approximate posterior draws \(\{(p^{(r)}, \tau^{(r)}): r=1,\dots,R\}\), possibly of different dimensions. We can subsequently reapply Bayes' formula conditional on \((p^{(r)},\tau^{(r)})\) to recover the posterior samples \(b^{(r)}\), which in turn can be used to construct posterior samples \(\mu^{(r)}\) of the mean function via~\eqref{eq:mu-h0}. Details on these calculations and the specific implementation of RJMCMC are given in Appendix~\ref{app:rjmcmc}.

\subsection{Decision Criteria} One of the main advantages of our Bayesian framework is that it naturally provides several posterior metrics to evaluate the evidence against the null hypothesis given the data. As we said before, the posterior probability of the null, \(\Pi_n(b=0|X_{1:n})\), corresponds to \(\E[\mathds{1}\{b=0\}|X_{1:n}]\), where \(\mathds{1}\{\cdot\}\) is the indicator function. We could naively compute the estimator \(R^{-1}\sum_{r} \mathds{1}\{b^{(r)}=0\}\) with our posterior samples, which by the ergodic properties of the RJMCMC sampler would indeed converge almost surely to the target probability (under standard assumptions). However, this estimator introduces the Monte Carlo noise of actually sampling \(b^{(r)}\), which is not desirable. Instead, we can consider the conditional posterior probability of the spike (computed through Bayes' formula)
\begin{equation}\label{eq:psi-posterior-spike}
\psi_n(p,\tau) = \Pi_n(b=0|p,\tau, X_{1:n}) = \frac{\pi_0}{m(X_{1:n}|p,\tau)},
\end{equation}
and use the equivalence
\[
 \E[\psi_n(p,\tau)|X_{1:n}] = \sum_p \int_{\mathcal T_p} \psi_n(p,\tau) \Pi_n(p,\tau|X_{1:n})\, d\tau = \Pi_n(b=0|X_{1:n})
\]
to obtain a more efficient estimator of the posterior probability of \(H_0\). Indeed, by the law of total variance we have:
\[
\var(\psi_n(p,\tau)|X_{1:n}) = \var(\E[\mathds{1}\{b=0\}|p, \tau, X_{1:n}]|X_{1:n}) \leq \var(\mathds{1}\{b=0\}|X_{1:n}).
\]
Hence, conditioning on \((p,\tau)\) and averaging the posterior spike probabilities \(\psi_n(p^{(r)}, \tau^{(r)})\) instead of the raw indicators \(\mathds{1}\{b^{(r)}=0\}\) yields the improved estimator\footnote{This strategy is similar to the Rao--Blackwellization process to reduce the variance of an estimator \citep[Theorem 7.9]{ghosh2006introduction}.}
\[
\hat \pi_{\text{null}} = \frac{1}{R}\sum_{r=1}^R \psi_n(p^{(r)}, \tau^{(r)}).
\]
Under the \(0\)--\(1\) loss function, we would reject the null hypothesis if \(\hat \pi_{\text{null}} \leq 1/2\) and accept it otherwise, but practitioners can fix any arbitrary threshold on a case-by-case basis.

Another quantity of interest to assess the evidence against the null hypothesis is the Bayes factor, which is defined as
\[
\bff_{10}(X_{1:n}) = \frac{L(X_{1:n}|H_1)}{L(X_{1:n}|H_0)} = \frac{1 - \Pi_n(H_0|X_{1:n})}{\Pi_n(H_0|X_{1:n})} \cdot \frac{\pi_0}{1-\pi_0}
\]
and represents the evidence in favor of \(H_1\) over \(H_0\) as supported by the data. For example, if \(\pi_0=1-\pi_0=1/2\), then \(\Pi_n(H_0|X_{1:n})=1/2\) corresponds to \(\bff_{10}(X_{1:n}) = 1\). Note that in our model \(L(X_{1:n}|H_0)\) is actually equal to 1 by construction, since we take \(\pp_0\) as the reference measure. In practice, we report an estimated Bayes factor 
\[
\hat{\bff}_{10} = \frac{1 - \hat \pi_{\text{null}}}{\hat \pi_{\text{null}}}\cdot \frac{\pi_0}{1-\pi_0}
\]
and use Jeffreys' scale (see Table~\ref{tab:jeffreys-scale}) to weigh the evidence against \(H_0\). Observe that the Bayes factor (a likelihood ratio) is actually independent of the prior belief \(\pi_0\), so it acts as a somewhat more objective measure to compare evidence against \(H_0\) across different models.

\begin{table}[t]
  \centering
  \setlength{\tabcolsep}{7pt}
  \renewcommand{\arraystretch}{1.05}
  \begin{tabular}{@{}rrl@{}}
  \midrule
   \multicolumn{1}{c}{\(\bff_{10}\)} & \multicolumn{1}{c}{\(\log \bff_{10}\)} & \multicolumn{1}{c@{}}{Strength of evidence against \(H_0\)} \\
  \midrule
  \(<1\) & \(<0\) & Negative\\
  1 to 3 & 0 to 1.1 &  Barely worth mentioning \\
  3 to 20 & 1.1 to 3 & Positive \\
  20 to 150 & 3 to 5 & Strong \\
  \(> 150\) &  \(>5\) & Very strong \\
  \midrule
  \end{tabular}
  \caption{Jeffreys' scale for the Bayes factor, adapted by \citet{kass1995bayes}.}\label{tab:jeffreys-scale}
\end{table}

Lastly, we can report the posterior probability of the so-called region of practical equivalence (ROPE), defined for a small \(\epsilon > 0\) as 
\[
\Pi_n\left(\sup_s |\mu(s)| < \epsilon \mathrel{\big\vert} X_{1:n}\right).
\]
This ``practical null'' probability is sometimes used to evaluate \(H_0\) when an exact zero is not meaningful, and is readily available in our Bayesian setting without changing the overall approach (i.e.\ without considering a precise hypothesis \(H_0^\epsilon:\sup_s|\mu(s)| < \epsilon\) from the start). The easiest way to approximate the ROPE probability on the observation grid \(\mathcal S\) is through the Monte Carlo estimator
\[
\hat\pi_{\text{ROPE}}(\epsilon) = \frac{1}{R}\sum_{r=1}^R \mathds{1}\left\{\max_{s\in\mathcal S} |\mu^{(r)}(s)| < \epsilon\right\}.
\]

\subsection{Posterior Summaries and Uncertainty Quantification}\label{sec:posterior-summaries}

An additional benefit of the Bayesian framework is that it facilitates posterior inference on the mean function as a by-product of the testing procedure. For example, we can estimate the posterior mean function \(\mu^{\text{PM}}(\cdot)=\E[\mu(\cdot)|X]\) using the conditional expectation strategy again, which is typically much less noisy than averaging posterior draws \(\mu^{(r)}\) directly. Indeed, for a fixed \(s\in[0,1]\) the law of total expectation gives \(\E[\mu(s)|X] = \E[\E[\mu(s)|p, \tau, X]|X]\), and the same variance reduction argument as before applies. In the i.i.d.\ case, the inner posterior expectation admits the closed-form expression 
\[
\E[\mu(s)|p, \tau, X_{1:n}] = (1-\psi_n(p, \tau))K(\tau, s)'\Lambda_\tau^{-1}S_n(\tau),
\]
where \(K(\tau, s)\) denotes the \(p\times 1\) vector with entries \(K(t_j, s)\), and \(\Lambda_\tau^{-1}S_n(\tau)\) is the posterior expectation of the coefficients under the slab (see Appendix~\ref{app:rjmcmc}). Consequently, a Monte Carlo estimator of the posterior mean function is given by
\[
\hat\mu_n^{\text{PM}}(\cdot) = \hat{\E}[\mu(\cdot)|X_{1:n}]=\frac{1}{R} \sum_{r=1}^R (1-\psi_n(p^{(r)}, \tau^{(r)}))K(\tau^{(r)}, \cdot)'\Lambda_{\tau^{(r)}}^{-1}S_n(\tau^{(r)}),
\]
which is then evaluated on the grid \(\mathcal S\). Note that this can be seen as a regularized estimator for \(\mu\) in which a certain weight \(\pi_0\) was placed on it being the zero function.

Along with the posterior expectation of \(\mu\), it is common to report credible bands around it to quantify the uncertainty in the posterior estimation. In particular, for any \(\alpha\in(0,1)\) we compute simultaneous \(1-\alpha\) credible bands, that is, bands that contain the entire function \(\mu\) with posterior probability \(1-\alpha\). The usual approach for constructing such bands is the maximum deviation method described by \citet{crainiceanu2007spatially}. The main idea is to standardize the deviations from the posterior mean by the pointwise posterior standard deviation \(\sigma(s)=\sqrt{\var(\mu(s)|X_{1:n})}\), to obtain bands that adapt to the local uncertainty. In particular, we aim to find a single critical value \(q_{1-\alpha}\) satisfying
\[
\Pi_n\left(\sup_{s\in [0,1]}\, \frac{|\mu(s) - \mu^{\text{PM}}(s)|}{\sigma(s)} \leq q_{1-\alpha}| X_{1:n}\right)= 1-\alpha,
\]
and then form symmetrical bands \(\{\mu^{\text{PM}}(s) \pm q_{1-\alpha}\sigma(s): s\in [0,1]\}\). In practice, this procedure is carried out on the discretized grid. First, compute the pointwise posterior standard deviation at each grid point \(s\in\mathcal S\) as 
\[
\hat\sigma(s)=\sqrt{\var(\{\mu^{(r)}(s)\}_{r=1}^R)}.
\]
Then, for each draw compute the standardized supremum deviation
\[
D^{(r)} = \max_{s\in\mathcal S} \frac{|\mu^{(r)}(s) - \hat\mu_n^{\text{PM}}(s)|}{\hat\sigma(s)}.
\]
Lastly, find the empirical \(1-\alpha\) quantile \(\hat q_{1-\alpha}\) of \(\{D^{(r)}\}_{r=1}^R\), and construct the credible bands \(\{\hat\mu_n^{\text{PM}}(s) \pm \hat q_{1-\alpha}\,\hat\sigma(s): s \in \mathcal S\}\). Depending on whether the zero function lies entirely within the band, we have another metric to reject \(H_0\) (at probability level \(1-\alpha\)).

\subsection{Estimation of the Covariance Function}\label{sec:covariance-estimation} 

So far we have treated \(K\) as known, but in practice we rarely have access to the complete covariance function. The FDA literature offers a variety of estimation techniques depending on the application and available data, such as maximum likelihood, functional principal component analysis (FPCA), and smooth interpolation methods \citep[e.g.][]{li2010uniform, xiao2018fast, wong2019nonparametric,mbaka2025estimating}. Parametric approaches---both frequentist and Bayesian---are also widely used, but they often require extra assumptions that lead to more complex models. We advocate for a simpler plug-in strategy in which an estimated covariance is computed from the data, and then used in our model as if it were the true covariance. The performance of the plug-in test for finite sample sizes will depend on how well the estimated covariance approximates \(K\), but as we will see in Section~\ref{sec:consistency-plug-in}, the asymptotic properties of the test are unaffected if the estimator itself is uniformly consistent.

Denote the sample mean by \(\bar X_n(\cdot) = n^{-1}\sum_i X_i(\cdot)\) and consider the standard empirical covariance estimator
\[
\hat C_n(t,s) = \frac{1}{n-1}\sum_{i=1}^n (X_i(t)-\bar X_n(t))(X_i(s) - \bar X_n(s)).
\]
From a computational perspective we only need to consider its discretized counterpart \(\hat{\mathbf C}_n=[\hat{C}_n(s_g, s_h)]_{g,h=1}^G\), which approximates the discretized covariance \(\mathbf K=[K(s_g, s_h)]_{g,h=1}^G\). However, in functional settings the sample size is typically small relative to the number of grid points, so the empirical covariance can be highly rank-deficient. A common regularization technique consists on shrinking towards a well-conditioned target to lift the near-zero eigenvalues, improving the quality and stability of the estimator. The Oracle Approximating Shrinkage (OAS) estimator \citep{chen2010shrinkage}, based on an earlier proposal by \citet{ledoit2004well} and improved for the Gaussian case, is a popular choice for this purpose. Its closed-form expression on the grid is
\[
\hat{\mathbf K}_n^{\text{OAS}} = (1-\hat\rho_n)\hat{\mathbf C}_n + \hat\rho_n \frac{\tr(\hat{\mathbf C}_n)}{G}I_G,
\]
where
\begin{equation}\label{eq:rho-oas}
\hat\rho_n = \min\left\{1,\, \frac{(1-2/G)\tr(\hat{\mathbf C}_n^2) + \tr^2(\hat{\mathbf C}_n)}{(n+1-2/G)\left[\tr(\hat{\mathbf C}_n^2) - \tr^2(\hat{\mathbf C}_n)/G \right]}\right\}.
\end{equation}

A convenient consequence of this formulation is that the resulting covariance matrix is strictly positive definite (except in pathological cases). Moreover, it can be shown that as the sample size grows the OAS estimator increasingly resembles the empirical covariance estimator \(\hat{\mathbf C}_n\), which is the expected behavior since there is less and less need for regularization.

\section{Asymptotic Theory}\label{sec:consistency}

To assess the frequentist behavior of the proposed Bayesian test, we study posterior consistency and posterior contraction rates. Posterior consistency means that, as the sample size grows, the posterior probability assigned to the true hypothesis converges to one (under the true data generating distribution), so the Bayesian procedure is asymptotically correct. More generally, a posterior contraction rate quantifies how fast the posterior mass concentrates around the true hypothesis as \(n\to\infty\). 

To evaluate contraction rates we will use the standard concept of \textit{stochastic boundedness} \citep[e.g.][Chapter 2]{vaart1998asymptotic}: given a sequence of nonnegative random variables \(Y_n\), whose distribution is determined by \(\pp\), and a sequence of positive numbers \(a_n\), we write \(Y_n=O_{\pp}(a_n)\) if for any \(\epsilon>0\) there exist \(N,M>0\) such that \(\pp(Y_n/a_n > M) < \epsilon\) for all \(n> N\). In this case, if \(a_n\to 0\) we say that \(a_n\) is a contraction rate for \(Y_n\), or equivalently, that \(Y_n\) decays at rate \(a_n\) (in probability).

We already established that the sample paths of the process \(X\) are supposed to be continuous, but we state this fact as an explicit assumption to emphasize its importance in the theoretical guarantees.

\begin{assumption}\label{ass:continuous-paths}
  The sample paths \(t \mapsto X(t)\) are continuous almost surely.
\end{assumption}

As an abuse of notation, we still write \(\pp_0\) and \(\pp_\mu\) for the product measures \(\pp_0^{(n)}\) and \(\pp_{\mu}^{(n)}\) of \((X_1,X_2,\dots, X_n)\), respectively, when the dependence on the sample size \(n\) is clear from context. Moreover, we let \(\pp_0^\infty\) denote the distribution of \((X_1,X_2,\dots)\) under \(H_0\). The proofs of the results in this section are deferred to Appendix~\ref{app:proofs}.

\subsection{Consistency and Contraction Rates for Known Covariance}\label{sec:consistency-known}

We first focus on the case where \(H_0\) is true, that is, \(X_1,\dots,X_n \sim \gp(0, K)\) i.i.d.\ with distribution \(\pp_0\). As it is often the case in Bayesian tests with a point-null hypothesis, posterior consistency holds as long as \(\Pi(H_0)>0\), through a direct application of Doob's theorem \citep[][Theorem 10.24]{dass2004note,ghosal2017fundamentals}.

\begin{theorem}
  Under \(H_0\), \(\lim_{n\to\infty}\Pi_n(H_1|X_{1:n}) = 0\) a.s.\ \raisebox{.25pt}{\([\pp_0^\infty]\)}.
\end{theorem}

To study posterior contraction rates, we impose additional assumptions on the prior distributions.

\begin{assumption}\label{ass:compact-support-tau}
  The prior measure on the time instants, \(\Pi_\tau(\cdot|p)\), is supported on a compact subset \(\mathcal C_p\subset \mathcal T_p\), that is, \(\Pi_\tau(\mathcal C_p|p)=1\) for all \(p\) such that \(\Pi(p)>0\).
\end{assumption}
\begin{assumption}\label{ass:prior-p-finite}
  The prior \(\Pi(p)\) has finite support, i.e., \(\Pi(p) = 0\) for all \(p > P_{\max}\).
\end{assumption}

Note that~\ref{ass:compact-support-tau} is verified, for example, by enforcing a uniform minimum separation on the time points in \(\tau\) (which already happens in practice on the finite grid \(\mathcal S\)). Indeed, for a fixed \(\delta>0\) one can consider the compact set
\[
\mathcal C_p(\delta) = \{(t_1,\dots, t_p)\in [0,1]^p: t_1 \geq \delta,\ t_p \leq 1 - \delta,\ t_{j+1} - t_j \geq \delta\ \text{for}\ j=1,\dots,p-1\}.
\]
On the other hand,~\ref{ass:prior-p-finite} is effectively a truncation on the model dimension, but it is not as restrictive as it may seem. Firstly, it is a standard assumption in the literature, and secondly, our model is designed for simplicity and parsimony, so we in fact encourage fewer components. In any event, the unknown value of \(p\) can still vary over a set as large as needed (e.g.\ we can take \(P_{\max}=|\mathcal S|\)). With these assumptions in mind, we can show that under \(H_0\) the posterior probability of the alternative hypothesis decays polynomially in the sample size. Let \(p_{\min}=\min\{p\in \N: \Pi(p)>0\}\).

\begin{theorem}\label{th:consistency-rates-h0}
  Suppose~\ref{ass:continuous-paths}--\ref{ass:prior-p-finite} hold. Then, \(\Pi_n(H_1|X_{1:n}) = O_{\pp_0}(n^{-p_{\min}/2})\) under \(H_0\).
\end{theorem}

\begin{remark}
  If \(p_{\min}=1\) we recover the standard parametric rate \(n^{-1/2}\). Even though the asymptotic rate \(O_{\pp_0}(n^{-p_{\min}/2})\) would suggest that it is better to ignore low values of \(p\) in the prior to improve consistency speed, finite-sample performance can be hurt by considering many components, since they introduce increasingly high collinearity. Moreover, removing low-dimensional alternatives makes the test more conservative and reduces interpretability.
\end{remark}

We now turn to the case where \(H_1\) is true. Suppose that the model is well-specified and denote the true mean function by \(\mu_*\in \Hcal(K) \setminus \{0\}\), so that \(X_1,\dots,X_n \sim \gp(\mu_*, K)\) i.i.d.\ with distribution \(\pp_{\mu_*}\) (note that \(\mu_*\in\Hcal(K)\setminus \Hcal_0(K)\) is allowed). In this situation we need to make sure that the prior places positive mass on the true mean, so the support cannot be restricted.

\begin{assumption}\label{ass:positive-prior-mass}
  The prior distribution \(\Pi\) charges open sets, i.e., \(\Pi(p)>0\) for all \(p\in\N\) and \(\Pi_{\tau,b}^{\text{slab}}(U|p)>0\) for all open sets \(U \subseteq \mathcal T_p\times \R^p\), where \(\Pi_{\tau,b}^{\text{slab}}(\cdot|p)\) is the slab part of the joint prior measure on \((\tau, b)\).
  
\end{assumption}

This mild condition is already sufficient to derive exponential contraction rates (which immediately guarantees posterior consistency). The proof relies on the fact that \(\Hcal_0(K)\) is dense in \(\Hcal(K)\), so we can always find an approximation \(\mu_{\theta}\in\Hcal_0(K)\) as close as desired to \(\mu_*\) in the RKHS norm.

\begin{theorem}\label{th:consistency-rates-h1}
  Suppose~\ref{ass:continuous-paths} and~\ref{ass:positive-prior-mass} hold. Then, \(\Pi_n(H_0|X_{1:n}) = O_{\pp_{\mu_*}}(e^{-n\gamma})\) under \(H_1\), for any \(0 < \gamma < \frac{1}{2}\|\mu_*\|^2_K\).
\end{theorem}

\begin{remark}
  The upper bound on \(\gamma\) is not arbitrary; we will show in Section~\ref{sec:optimality} that it is in fact an information-theoretic limit. We also note that, as expected, the rate will be faster when it is easier to distinguish \(\mu_*\) from the zero function, that is, the larger \(\|\mu_*\|_K\) is.
\end{remark}

It is worth mentioning that~\ref{ass:positive-prior-mass} is a bit more restrictive than needed, and it cannot hold simultaneously with the assumptions of Theorem~\ref{th:consistency-rates-h0}. Looking carefully at the proof of Theorem~\ref{th:consistency-rates-h1} (see Appendix~\ref{app:proofs}), one notices that the approximated mean function \(\mu_\theta\) is only required to satisfy the condition \(\|\mu_* - \mu_{\theta}\|^2_K < \|\mu_*\|^2_K\). To make the result compatible with the practical assumptions~\ref{ass:compact-support-tau} and~\ref{ass:prior-p-finite}, we could replace~\ref{ass:positive-prior-mass} by the following condition.

\begin{assumptionalt}\label{ass:positive-prior-mass-alt}
  There exists \(\theta_0=(p_0,\tau_0,b_0)\) such that \(p_0\leq P_{\max}\), \(\tau_0\in \mathcal C_{p_0}\), \(b_0\in\R^{p_0} \setminus \{0\}\) and \(\|\mu_* - \mu_{\theta_0}\|^2_K < \|\mu_*\|^2_K\). Moreover, for each \(p\leq P_{\max}\) we have \(\Pi(p)>0\) and \(\Pi_{\tau,b}^{\text{slab}}(\cdot|p)\) has full support on \(\mathcal C_p \times \R^p\).
\end{assumptionalt}

The full support condition is trivially satisfied with a standard choice of priors. As for the approximation constraint, we claim that in practice it reduces to a very minor requirement on the grid \(\mathcal S\). Fix \(p_0=1\), consider any single time point \(s\) with \(\tau_0=(s)\in \mathcal C_1\), and set \(b_0=(\mu_*(s)/K(s,s))\), which is well-defined because \(K\) is strictly positive definite. Then, by expanding the RKHS norm and using the reproducing property, we get
\[
\|\mu_* - \mu_{\theta_0}\|^2_K = \|\mu_*\|^2_K + \|\mu_{\theta_0}\|^2_K - 2\dotprod{\mu_*}{\mu_{\theta_0}}_K = \|\mu_*\|^2_K - \frac{\mu_*(s)^2}{K(s,s)},
\]
which is smaller than \(\|\mu_*\|^2_K\) as long as \(\mu_*(s)\neq 0\). Thus,~\ref{ass:positive-prior-mass-alt} holds whenever there exists at least one point in \(\mathcal C_1\) (which is effectively equal to \(\mathcal S\)) where \(\mu_*\) does not vanish. 

We can state a variant of Theorem~\ref{th:consistency-rates-h1} under this alternative assumption, in which the exponential rate now includes an additional approximation slack.

\begin{theorem}\label{th:consistency-rates-h1-alt}
  Suppose~\ref{ass:continuous-paths}--\ref{ass:prior-p-finite} and~\ref{ass:positive-prior-mass-alt} hold. Then, \(\Pi_n(H_0|X_{1:n}) = O_{\pp_{\mu_*}}(e^{-n\gamma})\) under \(H_1\), for any \(0 < \gamma < \frac{1}{2}(\|\mu_*\|^2_K - \|\mu_* - \mu_{\theta_0}\|^2_K)\).
\end{theorem}

\begin{remark}
  By considering~\ref{ass:positive-prior-mass-alt} instead of~\ref{ass:positive-prior-mass} we gain a simpler condition that does not conflict with the assumptions of Theorem~\ref{th:consistency-rates-h0}, but the price we pay is a slower contraction rate under \(H_1\) (unless \(\mu_* \in \Hcal_0(K)\)). Moreover, even though~\ref{ass:positive-prior-mass-alt} can be satisfied with just \(p_0=1\), the rate improves with the quality of the approximation (i.e.\ when \(\|\mu_* - \mu_{\theta_0}\|^2_K\) is small), so we can expect much faster rates if the condition holds for bigger values of \(p_0\).
\end{remark}

\subsection{Consistency and Contraction Rates for Plug-in Covariance}\label{sec:consistency-plug-in}

We extend the results of Section~\ref{sec:consistency-known} to the case in which the true covariance function \(K(t,s)\) is unknown and replaced by a symmetric plug-in estimator \(\hat K_n(t,s)\), which need not be strictly positive definite. For any \(\tau\in\mathcal T_p\), let \(\hat K_{n,\tau} = [\hat K_n(t_j, t_k)]_{j,k=1}^p\) be the plug-in analog of the matrix \(K_\tau\) defined in~\eqref{eq:likelihood}. We will see that as long as we have uniform consistency of \(\hat K_{n,\tau}\) to \(K_\tau\), the approximation effect is negligible and the posterior contraction rates stay the same. Therefore, we impose the following condition.

\begin{assumption}\label{ass:uniform-consistency-plug-in}
  For every \(p\) with \(\Pi(p)>0\) and any compact set \(\mathcal C_p\subset \mathcal T_p\), it holds that
  \[
  \sup_{\tau\in \mathcal C_p}\|\hat K_{n,\tau} - K_\tau\| \stackrel{\pp}{\longrightarrow} 0.
  \]
\end{assumption}

It can be shown that both the empirical covariance estimator \(\hat C_n\) and the OAS estimator described in Section~\ref{sec:covariance-estimation} satisfy the stronger condition of almost sure uniform convergence. Consider for each \(\tau\in\mathcal T_p\) the discretized matrices
\[
\hat{C}_{n,\tau} = [\hat C_n(t_j, t_k)]_{j,k=1}^p \quad \text{and}\quad \hat{K}_{n,\tau}^{\text{OAS}} = (1-\hat\rho_{n,\tau})\hat C_{n,\tau} + \hat\rho_{n,\tau}\frac{\tr(\hat C_{n,\tau})}{p}I_p,
\]
where \(\hat\rho_{n,\tau}\in[0,1]\) is the OAS shrinkage coefficient computed via the analogue of~\eqref{eq:rho-oas} with \(\hat{\mathbf C}_n\) replaced by \(\hat{C}_{n,\tau}\) and \(G\) replaced by \(p\).

\begin{proposition}\label{prop:sample-cov-matrix-convergence}
  For any matrix norm \(\|\cdot\|\) and any compact set \(\mathcal C_p\subset \mathcal T_p\), we have:
  \begin{enumerate}
    \item[(i)] \(\|\hat {\mathbf C}_n - \mathbf K\| \to 0\) a.s.\ and \(\sup_{\tau\in \mathcal C_p} \|\hat C_{n,\tau} - K_\tau\| \to 0\) a.s.
    \item[(ii)] \(\|\hat{\mathbf K}^{\emph{OAS}}_n - \mathbf K\| \to 0\) a.s.\ and \(\sup_{\tau\in \mathcal C_p} \|\hat K^{\emph{OAS}}_{n,\tau} - K_\tau\| \to 0\) a.s.
  \end{enumerate}
\end{proposition}

This convergence holds under any \(\gp(\mu, K)\) distribution with continuous sample paths, so both estimators satisfy~\ref{ass:uniform-consistency-plug-in} simultaneously under \(\pp_0\) and \(\pp_{\mu_*}\). The following results simply state that the contraction rates for the plug-in posterior \(\hat \Pi_n(\cdot|X_{1:n})\) computed using a suitable \(\hat K_n\) do not degrade. The proofs are a refinement of the corresponding ones for the known covariance case, in which the plug-in error is identified and controlled. 

\begin{theorem}\label{th:consistency-rates-h0-plugin}
  Suppose~\ref{ass:continuous-paths}--\ref{ass:prior-p-finite} and~\ref{ass:uniform-consistency-plug-in} hold. Then, \(\hat \Pi_n(H_1|X_{1:n}) = O_{\pp_0}(n^{-p_{\min}/2})\) under \(H_0\).
\end{theorem}

\begin{theorem}\label{th:consistency-rates-h1-plugin}
  Suppose~\ref{ass:continuous-paths},~\ref{ass:positive-prior-mass} and~\ref{ass:uniform-consistency-plug-in} hold. Then, \(\hat \Pi_n(H_0|X_{1:n}) = O_{\pp_{\mu_*}}(e^{-n\gamma})\) under \(H_1\), for any \(0<\gamma<\frac{1}{2}\|\mu_*\|_K^2\). Alternatively, if~\ref{ass:continuous-paths}--\ref{ass:prior-p-finite},~\ref{ass:positive-prior-mass-alt} and~\ref{ass:uniform-consistency-plug-in} hold, then \(\hat \Pi_n(H_0|X_{1:n}) = O_{\pp_{\mu_*}}(e^{-n\gamma})\) under \(H_1\), for any \(0<\gamma<\frac{1}{2}(\|\mu_*\|_K^2 - \|\mu_* - \mu_{\theta_0}\|_K^2)\).
\end{theorem}

\subsection{On Optimality of the Contraction Rates}\label{sec:optimality}

An aspect of the theoretical exposition that might stand out is the asymmetry of the contraction rates: polynomial under \(H_0\) and exponential under \(H_1\). Intuitively, this disparity can be traced to the spike-and-slab prior~\eqref{eq:prior-b} placed on the coefficient vector \(b\). Because the slab prior assigns nonnegligible mass arbitrarily close to the null region \(\{b=0\}\), under \(H_0\) the model struggles to rule out small nonzero mean functions, thereby slowing posterior concentration to the null. In contrast, under \(H_1\) the spike prior---a point mass at zero---is well separated from the truth, so the evidence accumulates much faster in favor of the alternative. This heuristic argument is made precise for example by \citet{johnson2010use}, who refer to such priors as ``local alternative priors''. Under some regularity conditions, they show that in this case the best achievable rate under a point-null is polynomial, whereas under the alternative the rate is typically exponential.

A natural follow-up question is whether the exponential rate under \(H_1\) can be improved, and the answer is negative. For a sequence of frequentist tests in the Neyman--Pearson sense (with bounded type I error probability) comparing a point-null against a simple alternative, the Chernoff--Stein lemma \citep[Theorem 11.8.3]{cover2006elements} states that the optimal decay rate of the type II error probability is exponential, with exponent constant equal to the Kullback--Leibler divergence between the underlying probability distributions. We show in the result below that this information-theoretic bound also applies to our Bayesian test, by thresholding the posterior to produce a frequentist test under the lemma's purview. Moreover, the Kullback--Leibler divergence in question is precisely \(\kl(\pp_{\mu_*} \,\|\, \pp_0)=\frac{1}{2}\|\mu_*\|_K^2\), so the rate established in Theorem~\ref{th:consistency-rates-h1} is arbitrarily close to the theoretical optimum.

\begin{proposition}\label{prop:optimality-rates-h1}
  There is no \(\zeta > \frac{1}{2}\|\mu_*\|_K^2\) such that \(\Pi_n(H_0|X_{1:n}) = O_{\pp_{\mu_*}}(e^{-n\zeta})\) under \(H_1\).
\end{proposition}

\section{Numerical Experiments}\label{sec:experiments}

We illustrate the finite-sample performance of the proposed Bayesian test through a series of experiments with simulated and real data. In both situations the covariance function is assumed unknown and estimated using the plug-in OAS estimator.

\subsection{Simulation Study} The simulated data comes from three different Gaussian processes (GP's):
\begin{enumerate}
  \item A Brownian motion (BM) with variance \(\sigma^2=1.5\).
  \item An Ornstein--Uhlenbeck process (OU) with variance \(\sigma^2=1.5\) and scale \(\ell=0.2\).
  \item A GP with a squared exponential kernel (SE), variance \(\sigma^2=1.5\) and scale \(\ell=0.2\).
\end{enumerate}
We draw \(n=50\) functional samples from each distribution, observed on a discretized grid of \(G=100\) equispaced points on \([0,1]\). The mean function under \(H_0\) is \(\mu_0\equiv 0\), and under \(H_1\) we consider two alternatives: a function in \(\Hcal_0(K)\) and a periodic function, given by
\[
\mu_*^{\text{RKHS}}(t) = 0.2K(0.15, t) - 0.3K(0.4, t) + 0.5K(0.75, t)\quad \text{and}\quad \mu_*^{\text{P}}(t) = 0.3\sin(2\pi t),
\]
respectively, where \(K\) is the corresponding covariance kernel in each case. We implement a reduced model with \(P_{\max}=5\) components, a uniform prior distribution \(p\sim\mathcal U\{1,2,3,4,5\}\), and a uniform prior for the time instants as well (restricted to the grid points in \(\mathcal S\)). Moreover, we assume for simplicity that the null and alternative hypotheses are a priori equally likely, so that \(\pi_0=1-\pi_0=1/2\). When choosing the crucial hyperparameter \(\eta\) one has to be wary of Lindley's paradox \citep{lindley1957statistical}. We propose to compute it from the discretized data through the formula
\[
\hat \eta_n=\eta_0 \sqrt{\frac{\tr(\hat{\mathbf K}_n^{\text{OAS}})}{G}}, \quad \eta_0>0,
\]
which measures the marginal standard deviation of the process, multiplied by a constant scale term. In this way we allow the prior on \(b\) to adapt to the scale of the data; we fix a value of \(\eta_0=1/4\) to allow detection of small signals (preliminary experiments suggest that bigger signals are easily recognized regardless of the value of \(\eta_0\)). Since \(\hat\eta_n\) converges in probability to a finite positive constant, this data-dependent choice does not affect the asymptotic properties of the test. Indeed, since \(\tr(\hat{\mathbf K}_n^{\text{OAS}})=\tr({\hat{\mathbf C}_n})\) and \(\tr(\hat{\mathbf C}_n) \to \tr({\mathbf K})\) by Proposition~\ref{prop:sample-cov-matrix-convergence}, we have \(\hat \eta_n \to \eta_\infty > 0\). We can refine the proofs in Appendix~\ref{app:proofs} by restricting the arguments to the high-probability event \(\{\eta_\infty/2 \leq \hat\eta_n \leq 2 \eta_\infty\}\), on which \(\hat\eta_n\) is bounded away from zero and infinity. Since all quantities appearing in the proofs that depend on \(\eta\) do so in a continuous manner, the constants involved in the corresponding bounds can be chosen uniformly over the compact interval \([\eta_\infty/2, 2\eta_\infty]\), and the consistency arguments carry over without further modification.

Then, we compare our method (which we call RKBT) to other one-sample testing procedures in the literature:
\begin{itemize}
  \item {BON-T}. Pointwise \(t\)-test over all grid points with Bonferroni correction for multiple comparisons. This is the most naive approach, used only as a baseline.
  \item {\(L^2\)}. Global test based on the squared \(L^2\)-norm \(\|\Delta_n\|_2^2 = n\int_0^1 \bar{X}_n^2(t)\, dt\), whose null distribution is a \(\chi^2\)-type mixture.
  \item {\(L^2\)-FPCA}. A version of the \(L^2\) test with projections through FPCA.
  \item {\(F\)}. Global test based on the regularized statistic \(F_n=\|\Delta_n\|_2^2/\tr(\hat{K}_n)\), which follows an \(F\)-type mixture distribution under \(H_0\).
  \item {BAYES-FPCA}. Standard Bayesian spike-and-slab test with dimension reduction via FPCA and a conjugate \(g\)-prior under \(H_1\). 
\end{itemize}
The frequentist competitors are primarily those outlined by \citet{zhang2013analysis}, whereas the Bayesian competitor arguably represents a novel---albeit simple---approach. A detailed breakdown of the comparison methods is provided in Appendix~\ref{app:experiments}.

Performance is measured through several frequentist metrics, repeating the whole procedure 100 times with different random seeds. We compare the rejection rate under \(H_0\) (type I error) and under \(H_1\) (power), as well as the AUC metric \citep{krzanowski2009roc}, which is a number in \([0,1]\) that measures the overall ability to correctly distinguish between positive and negative cases across all possible decision thresholds. The frequentist methods are calibrated to the nominal size of \(\alpha=0.05\), and the Bayesian methods are thresholded through the estimated Bayes factor to produce a binary accept/reject decision. Specifically, we use the standard rejection region provided by Jeffreys' scale (see Table~\ref{tab:jeffreys-scale}), which is given by \(\{\hat{\bff}_{10}>3\}\) and corresponds to \(\{\hat\pi_{\text{null}}<1/4\}\) with our prior assumptions. 

The results of the simulation study are presented in Table~\ref{tab:metrics-rkhs} and Table~\ref{tab:metrics-p}. We can see that our RKBT method has great control of the type I error, while also achieving decent power and AUC in most cases, despite having only a relatively low number of components to work with. Moreover, the comparison methods that obtain a higher power generally do so at the expense of a greater type I error, which may be undesirable. Note that we could increase the power by lowering the rejection threshold, which is fairly conservative, to a point where we still maintain an acceptable type I error (i.e.\ below 5\%). The execution time of our test, although not as fast as the competitors due to the computationally intensive RJMCMC sampling, is still adequate for practical use (see Appendix~\ref{app:experiments} for an explicit comparison).

\begin{table}[t]
  \small
  \centering
  \setlength{\tabcolsep}{7pt}
  \renewcommand{\arraystretch}{1.05}
  \begin{tabular}{@{}l rrr c rrr c rrr @{}}
  \midrule
  \multirow{2.5}{*}[0.5pt]{Method} & \multicolumn{3}{c}{BM} && \multicolumn{3}{c}{OU} && \multicolumn{3}{c@{}}{SE} \\
  \cmidrule(lr){2-4}
  \cmidrule(lr){6-8}
  \cmidrule(ll){10-12}
  & \multicolumn{1}{c}{T1} & \multicolumn{1}{c}{P} & \multicolumn{1}{c}{AUC} && \multicolumn{1}{c}{T1} & \multicolumn{1}{c}{P} & \multicolumn{1}{c}{AUC} && \multicolumn{1}{c}{T1} & \multicolumn{1}{c}{P} & \multicolumn{1}{c@{}}{AUC} \\
  \midrule
  RKBT & \textemph{0} & 67 & 90.9 && \textemph{0} & \textemph{88} & \textemph{98.5} && \textemph{0} & 84 & 98.4 \\
  \midrule
  BON-T & \textemph{0} & 35 & 81.4 && \textemph{0} & 71 & 97.2 && \textemph{0} & 67 & 92.9 \\
  \(L^2\) & 1 & \textemph{71} & \textemph{92.6} && 1 & 77 & 98 && 1 & 82 & 98.5 \\
  \(L^2\)-FPCA & 1 & 61 & 90.8 && 10 & 87 & 96.1 && 5 & 89 & 98.3 \\
  \(F\) & 1 & \textemph{71} & \textemph{92.6} && 1 & 77 & 98.2 && 1 & 81 & \textemph{98.7} \\
  BAYES-FPCA & 1 & 65 & 90.8 && 6 & 86 & 95.9 && 9 & \textemph{91} & 98.3 \\
  \midrule
  \end{tabular}
  \caption{Type I error (T1) under \(H_0\), power (P) under the true mean function \(\mu_*^{\text{RKHS}}\), and corresponding AUC, averaged over 100 independent replications for the three synthetic data sets (BM, OU, SE). All metrics are reported as percentages (\(\times 10^2\)), and the best results within each column are highlighted.}\label{tab:metrics-rkhs}
\end{table}

\begin{table}[t]
  \small
  \centering
  \setlength{\tabcolsep}{7pt}
  \renewcommand{\arraystretch}{1.05}
  \begin{tabular}{@{}l rrr c rrr c rrr @{}}
  \midrule
  \multirow{2.5}{*}[0.5pt]{Method} & \multicolumn{3}{c}{BM} && \multicolumn{3}{c}{OU} && \multicolumn{3}{c@{}}{SE} \\
  \cmidrule(lr){2-4}
  \cmidrule(lr){6-8}
  \cmidrule(ll){10-12}
  & \multicolumn{1}{c}{T1} & \multicolumn{1}{c}{P} & \multicolumn{1}{c}{AUC} && \multicolumn{1}{c}{T1} & \multicolumn{1}{c}{P} & \multicolumn{1}{c}{AUC} && \multicolumn{1}{c}{T1} & \multicolumn{1}{c}{P} & \multicolumn{1}{c@{}}{AUC} \\
  \midrule
  RKBT & \textemph{0} & 99 & \textemph{100} && \textemph{0} & 52 & 89 && \textemph{0} & 39 & \textemph{85.3} \\
  \midrule
  BON-T & \textemph{0} & 71 & 98.9 && \textemph{0} & 31 & 82.2 && \textemph{0} & 10 & 70.9 \\
  \(L^2\) & 1 & 58 & 97.4 && 1 & 52 & \textemph{90.2} && 1 & 43 & 84.6 \\
  \(L^2\)-FPCA & 1 & \textemph{100} & \textemph{100} && 10 & \textemph{57} & 83.3 && 5 & 47 & 83.1 \\
  \(F\) & 1 & 57 & 97.4 && 1 & 53 & \textemph{90.2} && 1 & 41 & 84.7 \\
  BAYES-FPCA & 1 & \textemph{100} & \textemph{100} && 6 & 53 & 83.2 && 9& \textemph{63} & 83.6 \\
  \midrule
  \end{tabular}
  \caption{Type I error (T1) under \(H_0\), power (P) under the true mean function \(\mu_*^{\text{P}}\), and corresponding AUC, averaged over 100 independent replications for the three synthetic data sets (BM, OU, SE). All metrics are reported as percentages (\(\times 10^2\)), and the best results within each column are highlighted.}\label{tab:metrics-p}
\end{table}

\subsection{Real Data Application}

The metrics in the simulation study are frequentist in nature, but the strength of our Bayesian test lies in the interpretability and the ability to quantify uncertainty. To analyze this aspect, we present a brief case study with real climate data. We consider the average monthly temperature recorded by 24 weather stations in Australia from 1990 to 2019, grouping the curves into three decades: the 1990s (1990--1999), the 2000s (2000--2009), and the 2010s (2010--2019).\footnote{The data is available through the \textsc{Meteostat} Python library~\citep{lamprecht2026meteostat}.} For each station and each decade, we compute a 12-point temperature curve by averaging temperatures within each month across all years in that decade. We then form pairwise differences between the decade-specific curves and apply our test to assess the null hypothesis of no difference in the mean temperature. We use the same prior distributions as in the simulations, but we allow \(P_{\max}=12\) and set \(\pi_0=1/3\), to convey the widespread belief that the temperature is actually increasing globally.

The Bayesian metrics obtained are summarized in Table~\ref{tab:metrics-climate}. Using the same rejection region as in the simulation study, namely \(\{\hat{\bff}_{10}>3\}\), our test decisively rejects \(H_0\) for the second and third period comparisons. This is consistent with the general observation that the climate is getting warmer in recent times. We can also see, for example, that the model gives very low posterior probability to differences of less than \(0.1^\circ\mathrm{C}\) from the 2010s to the 1990s, signalling that the discrepancy is bigger in that case. Note that around five components (months) seem to be sufficient to model the temperature differences, supporting our parsimony assumptions.

\begin{table}[t]
  \centering
  \setlength{\tabcolsep}{7pt}
  \renewcommand{\arraystretch}{1.05}
  \begin{tabular}{@{}lrrrr@{}}
  \midrule
  Period comparison & \multicolumn{1}{c}{\(\hat\pi_{\text{null}}\)} & \multicolumn{1}{c}{\(\hat{\bff}_{10}\)} & \multicolumn{1}{c}{\(\hat{\E}[p|X_{1:n}]\)} & \multicolumn{1}{c@{}}{\(\hat{\pi}_{\text{ROPE}}(\epsilon=0.1^\circ\mathrm{C})\)} \\
  \midrule
  2000s-1990s & 0.3302 & 1.0142 & 4.6807 & 0.4210 \\
  2010s-2000s & 0.0635 & 7.3709 & 5.2309 & 0.2120 \\
  2010s-1990s & 0.0784 & 5.8815 & 4.8832 & 0.0980 \\
  \midrule
  \end{tabular}
  \caption{Estimated posterior metrics given by our RKBT test for the three pairwise-difference temperature data sets: the posterior probability of the null, the Bayes factor, the posterior mean of the number of components in the model, and the ROPE probability \(\Pi_n(\sup_s |\mu(s)| < 0.1|X_{1:n})\).}\label{tab:metrics-climate}
\end{table}

Lastly, Figure~\ref{fig:climate_posterior} provides a useful functional interpretation of the posterior evidence reported in Table~\ref{tab:metrics-climate}, illustrating the posterior mean function \(\hat\mu_n^{\text{PM}}\) and the 95\% simultaneous bands described in Section~\ref{sec:posterior-summaries}. The estimated mean difference is relatively small and close to zero for the 2000s--1990s comparison, whereas the comparisons involving the 2010s display a markedly positive posterior mean over much of the year, with the largest discrepancy occurring between the 2010s and 1990s. Observe that the variability in the credible bands is generally more pronounced in the warmer months (September--March in the Southern Hemisphere), indicating that the uncertainty is higher precisely where the temperature differences are largest.

\begin{figure}[t]
  \centering
  \includegraphics[width=.99\textwidth]{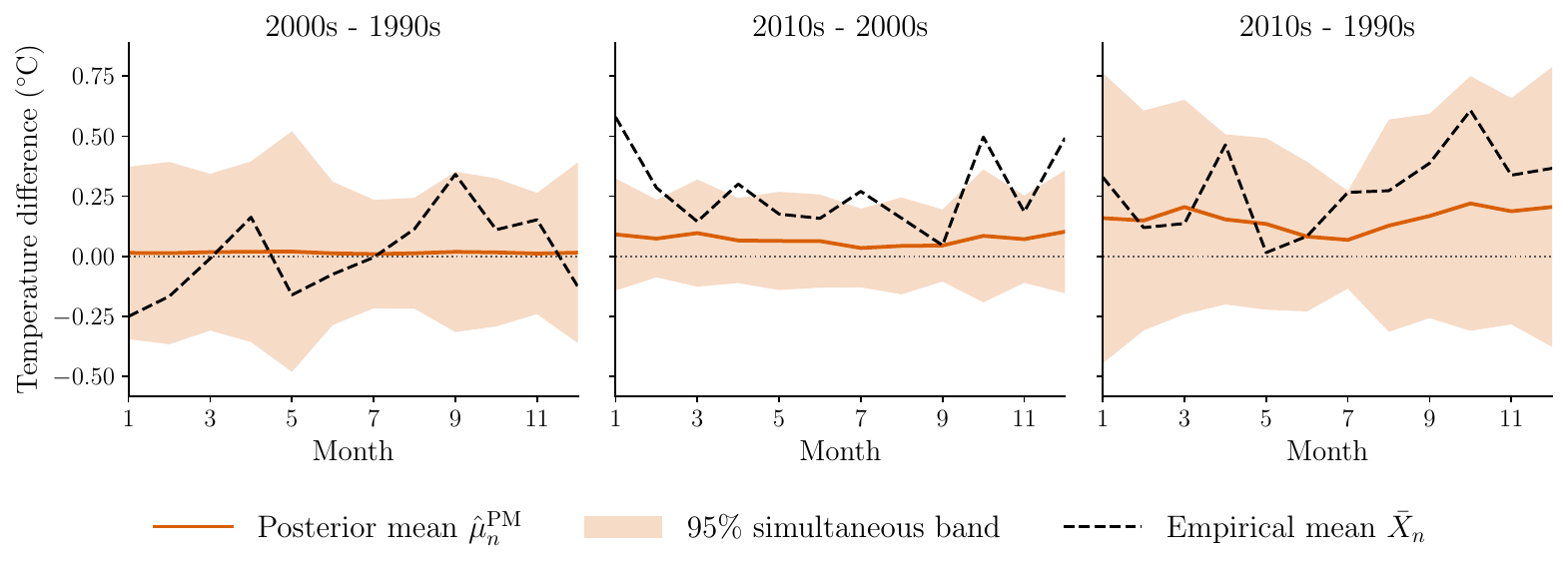}
  \caption{Estimated posterior mean and 95\% simultaneous credible bands given by our RKBT test for the three pairwise-difference temperature data sets. The empirical mean is also shown for reference.}\label{fig:climate_posterior}
\end{figure}

\section{Discussion}\label{sec:conclusions}

In this work we have described a Bayesian test for the mean \(\mu\) of a Gaussian process, observed through i.i.d.\ functional samples. The proposed model hinges on a finite-dimensional representation enabled by the reasonable assumption that \(\mu\in\Hcal(K)\), where \(\Hcal(K)\) is the RKHS associated with the covariance function \(K\) of the process. The case \(\mu\notin \Hcal(K)\), though not treated here, is also interesting in its own right. When the true mean function \(\mu_*\) lies outside \(\Hcal(K)\), by the Feldman--Hájek dichotomy the measures \(\pp_{\mu_*}\) and \(\pp_0\) are mutually singular, so there exists a Borel set \(S\) such that \(\pp_{\mu_*}(S)=1\) and \(\pp_0(S)=0\). In principle, we could define the test \(\phi(X)=\mathds{1}\{X\in S\}\), which has both type I and type II errors equal to zero, and can separate \(\pp_{\mu_*}\) and \(\pp_0\) perfectly with a single full observation of the process. This is an example of a ``near-perfect'' testing situation, in which the functional nature of the data allows exact verification of the null hypothesis in theory. In a discretized setting this effect might not be detectable with a finite grid, but it can manifest as the grid becomes dense \citep{torrecilla2020optimal}. Other instances of this phenomenon applied to functional data classification are studied, for example, by \citet{delaigle2012achieving} or \citet{berrendero2018use}. 

The Bayesian formulation of the test is flexible enough to accommodate a range of priors, fixing only the so-called ``local'' spike-and-slab prior on the coefficient vector \(b\), which is a common choice in multivariate settings. Other authors \citep[e.g.][]{johnson2010use} have objected to local alternative priors on a philosophical level, proposing different nonlocal alternatives. However, adopting them in our context would render the closed-form marginalization of \(b\) invalid, which is central to our computational strategy. Our approach connects the functional testing problem to a variable-dimension model in a principled way, and is theoretically backed by near-optimal posterior consistency results. These findings confirm that the Bayes factor-based test is asymptotically reliable under both the null and the alternative, and they provide a more refined description of how quickly evidence accumulates in favor of the correct hypothesis. This analysis is complemented by numerical studies with a tractable computational implementation, showing that our test is on a par with other methods in the literature. Lastly, we have also illustrated how the Bayesian aspect of the model can be used to quantify the uncertainty in the posterior distribution and interpret the results of the test, since we have access to a continuous posterior probability that measures the strength of evidence against the null.

We believe that this Bayesian framework opens several compelling paths for future research, demonstrating the practical and theoretical utility of RKHS representations in managing the inherent infinite-dimensionality of functional data. In particular, exploring the role of impact points within the RKHS approximation could yield meaningful new methodologies for change-point detection or variable selection in functional estimation. Furthermore, to our knowledge, the trans-dimensional aspect of the posterior approximation through RJMCMC---together with previous work by \citet{berrendero2025bayesian}---is a novel contribution in functional data.


\acks{The authors would like to acknowledge the valuable computational resources provided by the Centro de Computación Científica--Universidad Autónoma de Madrid (CCC--UAM). 
This research was partially supported by grants PID2023-148081NB-I00 and PRE2020-095147 of the Spanish Ministry of Science and Innovation (MCIN), co-financed by the European Social Fund (ESF). The authors also acknowledge financial support from grant CEX2023-001347-S funded by MICIU/AEI/10.13039/501100011033.}

\appendix
\section{RJMCMC Implementation}\label{app:rjmcmc}

The RJMCMC algorithm is used to jointly model the time instants \(\tau_p\) and the dimension \(p\) itself, constructing a Markov chain whose stationary distribution is the intended posterior distribution \(\Pi_n(p,\tau|X_{1:n})\) in~\eqref{eq:posterior-p-tau}. To generate new samples, the method considers \textit{within-model} moves (perturbing the current values of the time instants) and \textit{between-model} moves (increasing or decreasing the dimension). The former are implemented through minor modifications of the standard Metropolis--Hastings algorithm \citep{brooks2011handbook}, while the latter are generated by a reversible jump scheme \citep{richardson1997bayesian}. In our simplified setting, the acceptance probability of adding a new component \(t^*\) is given by
\[
\alpha_{\text{birth}} = \min\left\{1,\, \frac{\Pi(p+1) \Pi(\tau \cup \{t^*\}|p+1) m(X_{1:n}|p+1, \tau \cup \{t^*\})}{\Pi(p) \Pi(\tau|p) m(X_{1:n}|p, \tau)} \cdot \frac{r_{\text{delete}}(p)}{r_{\text{select}}(p)\cdot g(t^*)}\right\},
\]
where \(r_{\text{select}}(p)=1/(P_{\max} - p)\) is the probability of selecting a particular inactive component to become active, \(r_{\text{delete}}(p)=1/(p+1)\) is the reverse probability of deleting an existing component at random, and \(g(t^*)=1\) is the proposal density (uniform in our case). The acceptance probability of removing an existing component is computed similarly to maintain detailed balance, inverting the fractions in the expression above.

The RJMCMC sampling procedure is implemented in Python with the help of the \textsc{Eryn} library \citep{karnesis2023eryn}, which offers a flexible framework for RJMCMC based on ensembles of walkers (chains). Every walker in the ensemble maintains a set of up to \(P_{\max}\) active components, each storing one coordinate \(t_j\in\mathcal S\). For the experimentation in Section~\ref{sec:experiments} we deployed 64 walkers and ran each one for 2500 steps, discarding the first 1500 samples as burn-in (though convergence was typically achieved in fewer steps). Additional hyperparameters for the sampler were selected based on recommendations from the authors and a small set of trial runs. A more thorough description of the \textsc{Eryn} implementation and its internal mechanics is provided by \citet{berrendero2025bayesian}, along with an application to a similar FDA problem.

Regarding worst-case computational complexity, the RJMCMC procedure evaluates the slab marginal likelihood~\eqref{eq:m-slab} at each of the \(W\) walkers and \(R\) iterations. For a model with \(p\leq P_{\max}\) active components, this requires a Cholesky factorization of a \(p\times p\) matrix, which can be done in \(O(P_{\max}^3)\). Thus, including the initial \(O(nG^2)\) cost of the empirical covariance estimator (where \(G\) is the grid size), the overall worst-case time complexity of the RJMCMC sampling is \(O(WRP_{\max}^3+nG^2)\). Keeping a low \(P_{\max}\) and a sensible grid size allows the algorithm to run in reasonable time.

It is worth mentioning that in typical mixture-model settings, \textit{label switching} is a well-known identifiability issue: permuting the component labels \((t_j, \beta_j)\) yields the same likelihood, causing a multimodality in the posterior that can degrade mixing and complicate post-processing. In our case, however, label switching is not a practical concern for two reasons. First, as noted in Section~\ref{sec:rkhs-test}, we impose an ordering on the time instants, removing the symmetry and identifying a unique labeling. Second, and more importantly, most posterior quantities of interest are identifiable in the sense that they depend on the mean function parametrization \(\mu_\theta(\cdot) = \sum_{j} \beta_j K(t_j, \cdot)\), which is invariant under relabeling, so the functionals we report would be unchanged even if we did not impose an ordering.

\subsection{Recovering Posterior Samples for the Coefficients}

Once the RJMCMC sampling for \((p, \tau)\) is complete, we can generate posterior draws of \(b\) by applying Bayes' formula analytically (conditional on \(p\) and \(\tau\)). The corresponding posterior expression is
\[
\Pi_n(b|p,\tau, X_{1:n}) = \psi_n(p, \tau)\delta_0(b) + (1-\psi_n(p, \tau))\Pi_n^{\text{slab}}(b|p,\tau, X_{1:n}),
\]
where \(\psi_n(p, \tau)\) is the conditional posterior probability of the spike event, defined in~\eqref{eq:psi-posterior-spike}. For the slab part, we have \(\Pi_n^{\text{slab}}(b|p,\tau, X_{1:n}) \propto \Pi^{\text{slab}}(b)L(X_{1:n}|p, \tau, b)\), so that
\begin{align*}
\Pi_n^{\text{slab}}(b|p,\tau, X_{1:n}) &\propto \mathcal N_p(b|0, \eta^2 I_p) \exp\left\{S_n(\tau)'b - \frac{n}{2}b' K_{\tau} b\right\}\\
&\propto \exp\left\{ -\frac{1}{2}(b - \Lambda_\tau^{-1}S_n(\tau))'\Lambda_\tau(b -  \Lambda_\tau^{-1}S_n(\tau)) \right\}.
\end{align*}
Hence, \(b|p,\tau, X_{1:n}, \text{slab} \sim \mathcal N_p(\Lambda_\tau^{-1}S_n(\tau), \Lambda_\tau^{-1})\). Thus, after a full RJMCMC run with \(R\) posterior samples \((p^{(r)}, \tau^{(r)})\), we can retrieve the corresponding posterior samples \(b^{(r)}\) through the following simple procedure:
\begin{enumerate}
  \item Draw a spike or slab event \(Z^{(r)} \sim \text{Bernoulli}(\psi_n(p^{(r)}, \tau^{(r)}))\).
  \item If \(Z^{(r)}=1\), set \(b^{(r)}=0\). If \(Z^{(r)}=0\), sample \(b^{(r)}\) from \(\mathcal N_{p^{(r)}}(\Lambda_{\tau^{(r)}}^{-1}S_n(\tau^{(r)}), \Lambda_{\tau^{(r)}}^{-1})\).
\end{enumerate}

\section{Proofs}\label{app:proofs}

Here we present the proofs of the theoretical results in Section~\ref{sec:consistency}. Keep in mind that the covariance function \(K\) is strictly positive definite. Moreover, because the data are Gaussian, under~\ref{ass:continuous-paths} both \(\mu\) and \(K\) are continuous functions.\footnote{Gaussianity is not essential for the continuity of \(\mu\) and \(K\). Under~\ref{ass:continuous-paths}, finiteness of \(\E[\sup_t X(t)^2]\) suffices, which indeed holds automatically for Gaussian data \citep{landau1970supremum}.} To see this, let \(t_n\to t\) and \(u_n \to u\) in \([0,1]\). Since the sample paths are continuous a.s., \((X(t_n), X(u_n))\to (X(t), X(u))\) in distribution. By the Portmanteau lemma, the corresponding characteristic functions converge pointwise, and hence \(\mu(t_n)\to \mu(t)\) and \(K(t_n, u_n) \to K(t,u)\), because a Gaussian process is completely determined by its mean and covariance functions.

Using Bayes' rule~\eqref{eq:posterior-p-tau} for the posterior of \((p,\tau)\) and writing the normalization constant explicitly, the posterior probability of the alternative is given by
\begin{equation}\label{eq:posterior-slab}
\Pi_n(H_1|X_{1:n})= \frac{(1-\pi_0)m_1(X_{1:n})}{\pi_0 + (1-\pi_0)m_1(X_{1:n})},
\end{equation}
where
\begin{equation}\label{eq:posterior-slab-m1}
m_1(X_{1:n}) = \sum_{p\geq 1} \Pi(p)\int_{\mathcal T_p} \mslab(X_{1:n}|p,\tau)\Pi(\tau|p) \, d\tau.
\end{equation}

\subsection{Consistency and Contraction Rates for Known Covariance}

For notational convenience, define
\[
I_{n,p}(U) = \int_U \mslab(X_{1:n}|p, \tau)\Pi(\tau|p)\, d\tau, \quad \text{for every } U\subseteq \mathcal T_p.
\]

\subsubsection{Proof of Theorem~\ref{th:consistency-rates-h0}}

\begin{proof}
  Under~\ref{ass:compact-support-tau} and~\ref{ass:prior-p-finite}, the sum \(m_1(X_{1:n})\) simplifies to
  \[
  m_1(X_{1:n}) = \sum_{p=p_{\min}}^{P_{\max}} \Pi(p) I_{n,p}(\mathcal C_p).
  \]

  Fix \(p_{\min} \leq p\leq P_{\max}\) and let \(\epsilon>0\). Note that the minimum eigenvalue function \(\tau\mapsto\lambda_{\min}(K_\tau)\) is continuous by Weyl's inequality \citep[Theorem~VI.2.1]{bhatia1996matrix} and the continuity of \(K\). Because \(\mathcal C_p\) is compact and \(K_\tau\) is strictly positive definite, there exists \(\lambda_p>0\) with
  \[
  \lambda_{\min}(K_\tau)\ge \lambda_p>0\quad\text{for all }\tau\in \mathcal C_p.
  \]
  Since \(\Lambda_\tau = nK_\tau + \eta^{-2}I_p\), this implies that \(\lambda_{\min}(\Lambda_\tau) = n\lambda_{\min}(K_\tau) + \eta^{-2} \geq n\lambda_p>0\) for all \(\tau \in \mathcal C_p\). Therefore, the following holds for all \(\tau\in \mathcal C_p\):
  \begin{enumerate}
    \item[(i)] \(|\Lambda_\tau| \geq (n\lambda_p)^p\).
    \item[(ii)] \(v'\Lambda_{\tau}^{-1}v \leq \frac{1}{n\lambda_p}\|v\|^2\) for any \(v\neq 0\).
  \end{enumerate}
  Only the second assertion deserves a comment: the Rayleigh quotient of a symmetric matrix is bounded above by its maximum eigenvalue \citep[Theorem 4.2.2]{horn1985matrix}, and we have \(\lambda_{\max}(\Lambda_\tau^{-1})=(\lambda_{\min}(\Lambda_\tau))^{-1}\). Combining (i), (ii) and~\eqref{eq:m-slab} yields the upper bound
  \begin{equation}\label{eq:upper-bound-m-slab}
  \mslab(X_{1:n}|p,\tau) \leq \eta^{-p}(n\lambda_p)^{-p/2} \exp\left\{\frac{1}{2n\lambda_p}\|S_n(\tau)\|^2\right\}, \quad \text{for all } \tau \in \mathcal C_p.
  \end{equation}

  Now define \(Z_\tau = S_n(\tau)/\sqrt{n}\) for each \(\tau\), and note that \(Z_\tau \sim \mathcal N_p(0, K_\tau)\) under \(H_0\), so its distribution is independent of \(n\). Since
  \[
  \sup_{\tau \in \mathcal C_p} \frac{\|S_n(\tau)\|^2}{n} = \sup_{\tau \in \mathcal C_p} \|Z_\tau\|^2,
  \]
  we consider for any \(M>0\) the \(\pp_0\)-measurable event
  \begin{equation}\label{eq:event-E-n-M}
  E_{M} = \left\{\sup_{\tau\in \mathcal C_p} \|Z_\tau\|^2\leq M\right\}.
  \end{equation}
  The map \(\tau \mapsto Z_\tau(\omega)\) is continuous for almost every \(\omega\) by~\ref{ass:continuous-paths}, so on the compact set \(\mathcal C_p\) and the points of continuity in \(\omega\), the continuous function \(\tau \mapsto \|Z_\tau(\omega)\|^2\) attains its supremum. Hence, \(\sup_{\tau \in \mathcal C_p} \|Z_\tau\|^2 < \infty\) almost surely, which trivially implies that we can choose \(M\) large enough so that the complement \(E_{M}^c\) is negligible in probability, say \(\pp_0(E_{M}^c) < \epsilon\).

  On the other hand, on \(E_{M}\) the upper bound in~\eqref{eq:upper-bound-m-slab} reduces to the deterministic bound
  \begin{equation}\label{eq:deterministic-bound-sup-mslab}
  \sup_{\tau \in \mathcal C_p} \mslab(X_{1:n}|p,\tau) \leq \eta^{-p}(n\lambda_p)^{-p/2}\exp\left\{\frac{M}{2\lambda_p}\right\}.
  \end{equation}
  Since \(\Pi_\tau(\mathcal C_p|p)= 1\), by~\eqref{eq:deterministic-bound-sup-mslab} we have (still on the event \(E_{M}\))
  \[
  I_{n,p}(\mathcal C_p)\leq \sup_{\tau\in\mathcal C_p} \mslab(X_{1:n}|p,\tau) \int_{\mathcal C_p} \Pi(\tau|p)\, d\tau \leq Dn^{-p/2},
  \]
  where \(D>0\) is a constant independent of \(n\). Hence, for sufficiently large \(n\) we have
  \[
  \pp_0(I_{n,p}(\mathcal C_p) > Dn^{-p/2}) \leq \pp_0(E_{M}^c ) < \epsilon,
  \]
  so \(I_{n,p}(\mathcal C_p)=O_{\pp_0}(n^{-p/2})\). Write \(I_{n,p}(\mathcal C_p)=n^{-p/2}Y_{n,p}\), where \(Y_{n,p}=O_{\pp_0}(1)\) by definition. Then, because the dominant term in the finite sum \(m_1(X_{1:n})\) is the first term, we have
  \[
  m_1(X_{1:n}) = \sum_{p=p_{\min}}^{P_{\max}} \Pi(p) n^{-p/2}Y_{n,p} = O_{\pp_0}(n^{-p_{\min}/2}).
  \]
  Lastly, note that
  \[
  \Pi_n(H_1|X_{1:n}) \leq \frac{1-\pi_0}{\pi_0}m_1(X_{1:n}),
  \]
  so by the continuous mapping theorem the stochastic bound translates immediately and \(\Pi_n(H_1|X_{1:n})=O_{\pp_0}(n^{-p_{\min}/2})\), as claimed.
\end{proof}

\subsubsection{Proof of Theorem~\ref{th:consistency-rates-h1}}

\begin{proof}
  The RKHS framework will allow us to work, via a suitable approximation, in a finite-dimensional setting. First, since \(\Hcal_0(K)\) is dense in \(\Hcal(K)\), for any \(\epsilon>0\) there exists a parameter \(\theta_\epsilon = (p_\epsilon, \tau_\epsilon, b_\epsilon)\) such that
  \begin{equation}\label{eq:approximation-property-mu}
  \|\mu_* - \mu_{\tau_\epsilon, b_\epsilon}\|_K < \epsilon,
  \end{equation}
  where \(\mu_{\tau_\epsilon, b_\epsilon}\) is shorthand for \(\mu_{\theta_\epsilon}\). Since \(\mu_* \neq 0\), choosing \(\epsilon\) small enough (e.g.\ \(\epsilon\leq\|\mu_*\|_K\)) guarantees that \(\mu_{\tau_\epsilon, b_\epsilon}\neq 0\), and hence \(b_\epsilon \neq 0\). By~\ref{ass:positive-prior-mass}, \(\Pi(p_\epsilon)>0\) and there exists a neighborhood \(U =U_1 \times U_2\) of \((\tau_\epsilon, b_\epsilon)\) within \(\mathcal T_{p_\epsilon}\times \R^{p_\epsilon}\) with positive prior probability. We can (and do) take \(U_1\) and \(U_2\) compact, and arrange \(0\notin U_2\). 
  
  Let \((\tau, b)\in U\) and define \(\mu_*(\tau) = (\mu_*(t_1), \dots, \mu_*(t_{p_\epsilon}))'\) and \(Z_\tau = (S_n(\tau) - n\mu_*(\tau))/\sqrt{n}\). Note that \(Z_\tau\sim \mathcal N_{p_\epsilon}(0, K_\tau)\) under \(H_1\), so its distribution is independent of \(n\). Substituting in~\eqref{eq:likelihood} yields
  \begin{equation}\label{eq:log-likelihood-H1}
  \log L(X_{1:n}|p_\epsilon, \tau, b) = n\mu_*(\tau)'b - \frac{n}{2}b'K_\tau b + \sqrt{n}Z_\tau'b.
  \end{equation}
  On the one hand, for \((\tau, b)\in U\) define the function
  \[
  f(\tau, b) = \mu_*(\tau)'b - \frac{1}{2}b'K_\tau b = \mu_*(\tau)'b - \frac{1}{2}\|\mu_{\tau,b}\|^2_K,
  \]
  and observe that by the reproducing property and the linearity of the inner product we have
  \[
  \dotprod{\mu_*}{\mu_{\tau, b}}_K = \dotprod{\mu_*}{\ \sum_{j=1}^{p_\epsilon} \beta_j K(t_j, \cdot)}_K = \sum_{j=1}^{p_\epsilon} \beta_j \dotprod{\mu_*}{K(t_j, \cdot)}_K = \mu_*(\tau)'b.
  \]
  Thus, using the familiar identity \(\|f-g\|^2 = \|f\|^2 + \|g\|^2 - 2\dotprod{f}{g}\), we can write
  \[
  f(\tau, b) = \dotprod{\mu_*}{\mu_{\tau, b}}_K - \frac{1}{2}\|\mu_{\tau,b}\|^2_K = \frac{1}{2}\|\mu_*\|^2_K - \frac{1}{2}\|\mu_* - \mu_{\tau, b}\|^2_K.
  \]
  Moreover, by the triangle inequality we have
  \[
  \|\mu_* - \mu_{\tau, b}\|_K \leq \|\mu_* - \mu_{\tau_\epsilon, b_\epsilon}\|_K + \|\mu_{\tau_\epsilon, b_\epsilon} - \mu_{\tau,b}\|_K.
  \]
  Then, by~\eqref{eq:approximation-property-mu} and the continuity of the map \((\tau, b) \mapsto \mu_{\tau, b}\) in the RKHS norm (which follows from the continuity of \(K\)), after shrinking \(U\) if necessary we have \(\|\mu_* - \mu_{\tau,b}\|_K < 2\epsilon\) for all \((\tau,b)\in U\). Consequently,
  \[
  f(\tau, b) \geq \frac{1}{2}\|\mu_*\|^2_K - 2\epsilon^2 =: c_\epsilon > 0,
  \]
  where positivity holds by choosing \(\epsilon\) small enough. On the other hand, note that by the Cauchy--Schwarz inequality
  \[
  \sqrt{n} Z_\tau'b \geq \sqrt{n}\cdot \inf_{b\in U_2} Z_\tau'b \geq -\sqrt{n}\cdot \sup_{b\in U_2} \|b\|\cdot \sup_{\tau \in U_1}\|Z_\tau\|.
  \]

  All in all, we see from~\eqref{eq:log-likelihood-H1} that for all \((\tau, b)\in U\) we obtain a uniform lower bound for the log-likelihood given by
  \[
  \log L(X_{1:n}|p_\epsilon, \tau, b) = nf(\tau, b) +  \sqrt{n}Z_\tau'b \geq nc_\epsilon - \sqrt{n}B\cdot \sup_{\tau\in U_1}\|Z_\tau\|,
  \]
  where \(B=\sup_{b\in U_2} \|b\|\) satisfies \(0 < B < \infty\), since \(U_2\) is a compact set with \(0\notin U_2\). Hence, by~\eqref{eq:posterior-slab-m1} and~\eqref{eq:m-slab}, we have
  \begin{equation}\label{eq:lower-bound-m1}
    \begin{split}
    m_1(X_{1:n}) &\geq \Pi(p_\epsilon) \int_{U_1} \int_{U_2} \mathcal N_{p_\epsilon}(b|0, \eta^2 I_{p_\epsilon}) \exp\left\{nf(\tau, b) +  \sqrt{n}Z_\tau'b\right\}\Pi(\tau|p_{\epsilon})\, db\, d\tau \\
    &\geq D \exp\left\{nc_\epsilon - \sqrt{n}B\cdot \sup_{\tau\in U_1}\|Z_\tau\|\right\},
    \end{split}
  \end{equation}
  where \(D=\Pi(p_\epsilon)\Pi_\tau(U_1|p_{\epsilon})\Pi_b^{\text{slab}}(U_2|p_{\epsilon})>0\) by assumption. Next, note that \(\sup_{\tau\in U_1}\|Z_\tau\|\) is \(O_{\pp_{\mu_*}}(1)\) (by the same argument as in the proof of Theorem~\ref{th:consistency-rates-h0}), so for all \(\delta>0\) there exists \(M>0\) such that \(\pp_{\mu_*}(B\cdot \sup_{\tau\in U_1} \|Z_\tau\| \geq M) < \delta\). Thus, by~\eqref{eq:lower-bound-m1} we get
  \[
  \pp_{\mu_*}(m_1(X_{1:n}) \leq D \exp\{nc_\epsilon - \sqrt{n}M\}) \leq  \pp_{\mu_*}\left(B\cdot \sup_{\tau\in U_1} \|Z_\tau\| \geq M\right) < \delta.
  \]
  Since \(nc_\epsilon/(1+\epsilon)\) grows slower than \(nc_\epsilon - \sqrt{n}M\) as \(n\to\infty\), for sufficiently large \(n\) we get the probability bound
  \[
  \pp_{\mu_*}\left(m_1(X_{1:n}) \leq D \exp\{nc_\epsilon/(1+\epsilon)\}\right) \leq \pp_{\mu_*}\left(m_1(X_{1:n}) \leq D \exp\{nc_\epsilon - \sqrt{n}M\}\right) < \delta.
  \]
  To conclude, observe that by~\eqref{eq:posterior-slab}
  \[
  \Pi_n(H_0|X_{1:n}) = \frac{\pi_0}{\pi_0 + (1-\pi_0)m_1(X_{1:n})} \leq \frac{\pi_0}{(1-\pi_0)m_1(X_{1:n})},
  \]
  so if we define \(\tilde{D}=\pi_0/((1-\pi_0)D)\), for large \(n\) we have
  \[
  \pp_{\mu_*}\left(\Pi_n(H_0|X_{1:n}) > \tilde{D}\exp\left\{\frac{-nc_\epsilon}{1+\epsilon}\right\}\right) \leq \pp_{\mu_*}\left(m_1(X_{1:n}) \leq D \exp\left\{\frac{nc_\epsilon}{1+\epsilon}\right\}\right) < \delta,
  \]
  i.e., \(\Pi_n(H_0|X_{1:n}) = O_{\pp_{\mu_*}}(e^{-nc_\epsilon/(1+\epsilon)})\). Since \(\epsilon\) was arbitrarily small, the constant \(\gamma\) in the exponent is as close as desired to \(\lim_{\epsilon\to 0}\, c_\epsilon/(1+\epsilon)=\frac{1}{2}\|\mu_*\|^2_K\).
\end{proof}

\subsubsection{Proof of Theorem~\ref{th:consistency-rates-h1-alt}}

\begin{proof}
  By~\ref{ass:positive-prior-mass-alt}, there exist \(p_0\leq P_{\max}\), \(\tau_0\in \mathcal C_{p_0}\), and \(b_0\in\R^{p_0}\setminus\{0\}\) such that
  \[
  \xi:=\|\mu_*-\mu_{\tau_0,b_0}\|_K < \|\mu_*\|_K.
  \]
  By the full support condition in~\ref{ass:positive-prior-mass-alt}, there exists a compact neighborhood \(U=U_1\times U_2\) of \((\tau_0,b_0)\) within \(\mathcal C_{p_0}\times\R^{p_0}\setminus\{0\}\) with \(\Pi(p_0)\Pi_\tau(U_1|p_0)\Pi_b^{\text{slab}}(U_2|p_0)>0\). Let \((\tau,b)\in U\) and define \(f(\tau, b)\) as in the proof of Theorem~\ref{th:consistency-rates-h1}, i.e.,
  \[
  f(\tau, b) = \frac{1}{2}\|\mu_*\|_K^2 - \frac{1}{2}\|\mu_* - \mu_{\tau, b}\|_K^2.
  \]
  Let \(\epsilon >0\) be small enough so that \(\xi + \epsilon < \|\mu_*\|_K\). By shrinking \(U\) we can ensure that \(\|\mu_{\tau_0,b_0}-\mu_{\tau,b}\|_K < \epsilon\), and then by the triangle inequality we have \(\|\mu_*-\mu_{\tau,b}\|_K \leq \xi + \epsilon\). Therefore,
  \[
  f(\tau, b) \geq \frac{1}{2}\|\mu_*\|_K^2 - \frac{1}{2}(\xi + \epsilon)^2 =: c_{\epsilon,\xi} > 0 \quad\text{for all }(\tau,b)\in U,
  \]
  where positivity holds because \(\xi + \epsilon < \|\mu_*\|_K\). From this point on, we proceed exactly as in the proof of Theorem~\ref{th:consistency-rates-h1}, obtaining \(\Pi_n(H_0|X_{1:n})=O_{\pp_\mu*}(e^{-nc_{\epsilon,\xi}/(1+\epsilon)})\). Letting \(\epsilon\to 0\) gives a constant \(\gamma\) in the exponent as close as desired to \(\frac{1}{2}(\|\mu_*\|_K^2 - \xi^2)\).
\end{proof}

\subsection{Consistency and Contraction Rates for Plug-in Covariance}

Firstly, note that under~\ref{ass:continuous-paths} the empirical covariance estimator \(\hat C_n\) is uniformly consistent almost surely, i.e., 
\begin{equation}\label{eq:uniform-consistency-empirical-covariance}
\|\hat C_n - K\|_\infty \stackrel{\text{a.s.}}{\longrightarrow} 0\quad (n\to\infty),
\end{equation}
where \(\|\cdot\|_\infty\) is the supremum norm. This follows directly from Mourier's law of large numbers in Banach spaces \citep[Theorem 7.2.1]{laha1979probability} and the fact that \(\E[\|X\|_\infty^2] = \E[\sup_t |X(t)|^2] < \infty\). To establish the finiteness of the second moment under the supremum norm, \citet{landau1970supremum} showed that \(\E[e^{\epsilon \|X\|_\infty^2}] < \infty\) for sufficiently small \(\epsilon>0\), which immediately implies \(\E[\|X\|_\infty^2] < \infty\).

\subsubsection{Proof of Proposition~\ref{prop:sample-cov-matrix-convergence}}

\begin{proof}
  The first assertion in (i) is a consequence of the scalar law of large numbers applied entrywise. Alternatively, it can be immediately deduced from~\eqref{eq:uniform-consistency-empirical-covariance} by choosing the max-norm on matrices (since all matrix norms are equivalent):
  \[
  \|\hat{\mathbf C}_n - \mathbf K\|_{\max} = \max_{1\leq g,h\leq G} |\hat C_n(s_g, s_h) - K(s_g, s_h)| \leq \|\hat C_n - K\|_\infty \stackrel{\text{a.s.}}{\longrightarrow} 0.
  \]
  The second assertion is proven analogously.

  For (ii), note that \(\tr(\cdot)\) is continuous, so from (i) we have
  \[
  \tr(\hat{\mathbf C}_n) \stackrel{\text{a.s.}}{\longrightarrow} \tr(\mathbf K)\quad \text{and}\quad \tr(\hat{\mathbf C}_n^2) \stackrel{\text{a.s.}}{\longrightarrow} \tr(\mathbf K^2).
  \]
  Then, the numerator in \(\hat\rho_n\), defined in~\eqref{eq:rho-oas}, converges almost surely to a finite positive constant (as long as \(G\geq 2\)). For the denominator, observe that
  \begin{equation}\label{eq:rho-oas-denominator-convergence}
    \tr(\hat{\mathbf C}_n^2) - \frac{1}{G}\tr^2(\hat{\mathbf C}_n) \stackrel{\text{a.s.}}{\longrightarrow} \tr(\mathbf K^2) - \frac{1}{G}\tr^2(\mathbf K).
  \end{equation}
  By Cauchy--Schwarz (using the definition of the trace as the sum of the eigenvalues), the right-hand side of~\eqref{eq:rho-oas-denominator-convergence} is always nonnegative, and it vanishes if and only if the true covariance matrix \(\mathbf K\) is spherical. Thus, we consider two cases:
  \begin{itemize}
  \item If, say, \(\mathbf K =\sigma^2 I_G\) with \(\sigma^2>0\), then the behavior of \(\hat \rho_n\) as \(n\to\infty\) is not easily determined (though it can be analyzed by a central limit theorem argument). Luckily this is irrelevant, since now both the empirical covariance and target covariance matrices converge a.s.\ to the same limit, namely \(\mathbf K\). Then,
  \[
  \|\hat{\mathbf K}_n^{\text{OAS}} - \mathbf K\| \leq (1-\hat\rho_n) \|\hat{\mathbf C}_n - \mathbf K\| + \hat\rho_n \left\|\frac{\tr(\hat{\mathbf C}_n)}{G}I_G - \mathbf K\right\| \stackrel{\text{a.s.}}{\longrightarrow} 0.
  \]
  \item If \(\mathbf K\) is not spherical, then \(\|\tr(\hat{\mathbf C}_n)I_G/G - \mathbf K\|\) converges to a finite constant almost surely. On the other hand, the denominator in~\eqref{eq:rho-oas} diverges almost surely as \(n\to\infty\), so \(\hat\rho_n\stackrel{\text{a.s.}}{\to} 0\). Then, using the same upper bound as in the previous case, we also get \(\|\hat{\mathbf K}_n^{\text{OAS}} - \mathbf K\|\to 0\) a.s.
  \end{itemize}

  We now prove the second assertion in (ii) in detail. The case \(p=1\) is trivial, so we assume \(p>1\). By the triangle inequality and particularizing to the Frobenius norm \(\|\cdot\|_F\), we have
  \[
  \sup_{\tau\in \mathcal C_p}\|\hat K^{\text{OAS}}_{n,\tau} - K_\tau\|_F \leq \sup_{\tau\in \mathcal C_p}\|\hat C_{n,\tau} - K_\tau\|_F + \sup_{\tau\in \mathcal C_p}\hat\rho_{n,\tau}\left\|\frac{\tr(\hat C_{n,\tau})}{p}I_p - \hat C_{n,\tau}\right\|_F.
  \]
  The first term on the right-hand side goes to 0 almost surely by the uniform convergence established in (i). For the second term, let
  \[
  \hat{D}_{n,\tau} = \frac{\tr(\hat C_{n,\tau})}{p}I_p - \hat C_{n,\tau},
  \]
  and note that \(\|\hat{D}_{n,\tau}\|_F^2 = \tr(\hat D_{n,\tau}\hat D_{n,\tau}') = \tr(\hat C_{n,\tau}^2) - \tr^2(\hat C_{n,\tau})/p\). Then, for \(\tau\in \mathcal C_p\) we have
  \begin{equation}\label{eq:rho-frobenius-difference}
  \hat\rho_{n,\tau}\left\|\frac{\tr(\hat C_{n,\tau})}{p}I_p - \hat C_{n,\tau}\right\|_F = \min\left\{1,\, \frac{\hat{N}_{n,\tau}}{(n + 1 - 2/p)\|\hat{D}_{n,\tau}\|_F^2}\right\}\|\hat{D}_{n,\tau}\|_F,
  \end{equation}
  where \(\hat{N}_{n,\tau} = (1-2/p) \tr(\hat C_{n,\tau}^2) + \tr^2(\hat C_{n,\tau}) \geq 0\). Now, for any \(A\geq 0\), \(c>0\) and \(x\geq 0\) the algebraic inequality \(\min\left\{1,\frac{A}{cx^2}\right\}x \leq \sqrt{\frac{A}{c}}\) holds (if \(x \leq \sqrt{A/c}\) the bound is trivial since the minimum is at most \(1\); if \(x > \sqrt{A/c}\) it evaluates to \(\frac{A}{cx^2}\), and in turn \(\frac{A}{cx^2}x < \sqrt{A/c}\)). Applying this bound to~\eqref{eq:rho-frobenius-difference} yields
  \[
  \hat\rho_{n,\tau}\|\hat D_{n,\tau}\|_F \leq \sqrt{\frac{\hat{N}_{n,\tau}}{n + 1 - 2/p}}.
  \]

  To conclude, we show that \(\hat N_{n,\tau}\) is uniformly bounded almost surely for large \(n\). To see this, define the continuous map \(F:\R^{p\times p} \to \R\) given by \(F(A)= (1-2/p)\tr(A^2)+\tr^2(A)\). Note that the set \(\{K_\tau: \tau \in \mathcal C_p\}\) is compact, so uniform convergence of \(\hat C_{n,\tau}\) to \(K_\tau\), given by (i), implies that for large \(n\) all \(\hat C_{n,\tau}\) lie almost surely in a fixed compact neighborhood of this set. The function \(F\) is uniformly continuous on this compact neighborhood, and therefore 
  \[
  \sup_{\tau\in \mathcal C_p}|F(\hat C_{n,\tau}) - F(K_\tau)| = \sup_{\tau\in \mathcal C_p}|\hat N_{n,\tau} - F(K_\tau)|\stackrel{\text{a.s.}}{\longrightarrow} 0.
  \] 
  This implies that \(\sup_{\tau\in \mathcal C_p}\hat N_{n,\tau}\leq \sup_{\tau\in \mathcal C_p} F(K_\tau)+ 1 =: M\) a.s.\ for all \(n\) large enough, and \(M<\infty\) by continuity. Hence,
  \[
  \sup_{\tau \in \mathcal C_p} \hat\rho_{n,\tau}\|\hat D_{n,\tau}\|_F \leq \sqrt{\frac{M}{n + 1 - 2/p}} \stackrel{\text{a.s.}}{\longrightarrow} 0.
  \]
\end{proof}

We now prove the plug-in consistency results for our test. Define the plug-in marginal likelihood
\[
\hat \mslab(X_{1:n}|p,\tau) = \eta^{-p}|\hat\Lambda_{n,\tau}|^{-1/2}\exp\left\{\frac{1}{2}S_n(\tau)'\hat\Lambda_{n,\tau}^{-1}S_n(\tau)\right\},
\]
where \(\hat\Lambda_{n,\tau} = n\hat K_{n,\tau} + \eta^{-2}I_p \). Similarly, consider the plug-in posterior probability of the alternative
\[
\hat\Pi_n(H_1|X_{1:n}) = \frac{(1-\pi_0)\hat m_1(X_{1:n})}{\pi_0 + (1-\pi_0)\hat m_1(X_{1:n})},
\]
where
\[
\hat m_1(X_{1:n}) = \sum_{p\geq 1}\Pi(p)\int_{\mathcal T_p}\hat \mslab(X_{1:n}|p,\tau)\Pi(\tau|p)\,d\tau.
\]

\subsubsection{Proof of Theorem~\ref{th:consistency-rates-h0-plugin}}

\begin{proof}
  We will follow the proof of Theorem~\ref{th:consistency-rates-h0} and adapt it only where necessary. First, observe that by~\ref{ass:compact-support-tau} and~\ref{ass:prior-p-finite} the sum \(\hat m_1(X_{1:n})\) reduces to
\[
\hat m_1(X_{1:n}) = \sum_{p=p_{\min}}^{P_{\max}} \Pi(p)\,\hat I_{n,p}(\mathcal C_p),
\]
where \(\hat I_{n,p}(\mathcal C_p)=\int_{\mathcal C_p}\hat \mslab(X_{1:n}|p,\tau)\Pi(\tau|p)\,d\tau\). Fix \(p_{\min} \leq p\leq P_{\max}\), and recall that there exists \(\lambda_p>0\) with \(\lambda_{\min}(K_\tau)\geq \lambda_p\) for all \(\tau\in \mathcal C_p\). By Weyl's inequality, choosing the spectral norm \(\|\cdot\|_2\) on matrices, we have
\[
\lambda_{\min}(\hat K_{n,\tau}) \geq \lambda_{\min}(K_\tau) - \|\hat K_{n,\tau} - K_\tau\|_2 \geq \lambda_p - \sup_{\tau\in \mathcal C_p}\|\hat K_{n,\tau} - K_\tau\|_2.
\]
Define the event where \(\hat K_{n,\tau}\) is uniformly close to \(K_\tau\) with a suitable bound, i.e.,
\[
F_n = \left\{\sup_{\tau\in \mathcal C_p}\|\hat K_{n,\tau} - K_\tau\|_2 \leq \frac{\lambda_p}{2}\right\},
\]
and note that \(\pp_0(F_n^c)\to 0\) as \(n\to\infty\) by~\ref{ass:uniform-consistency-plug-in}. Furthermore, \(\lambda_{\min}(\hat K_{n,\tau})\geq \lambda_p/2\) for all \(\tau\in \mathcal C_p\) on \(F_n\), and consequently \(\lambda_{\min}(\hat\Lambda_{n,\tau})\geq n\lambda_p/2>0\). On \(F_n\), this yields the upper bound
\begin{equation}\label{eq:upper-bound-m-slab-plugin}
\hat \mslab(X_{1:n}|p,\tau) \leq \eta^{-p}\left(\frac{n\lambda_p}{2}\right)^{-p/2}\exp\left\{\frac{1}{n\lambda_p}\|S_n(\tau)\|^2\right\}, \quad \text{for all }\tau\in \mathcal C_p.
\end{equation}

Let \(\epsilon>0\) and define the event \(E_{M}\) for a large enough \(M\) as in~\eqref{eq:event-E-n-M}, so that \(\pp_0(E_M^c)<\epsilon/2\). On the high-probability event \(F_n\cap E_{M}\), incorporating~\eqref{eq:upper-bound-m-slab-plugin} gives after integration
\[
\hat I_{n,p}(\mathcal C_p) \leq \eta^{-p}\left(\frac{n\lambda_p}{2}\right)^{-p/2}\exp\left\{\frac{M}{\lambda_p}\right\} =: \hat D\, n^{-p/2},
\]
where \(\hat D=\eta^{-p}(\lambda_p/2)^{-p/2}\exp\{M/\lambda_p\}\) is a finite constant independent of \(n\). Then, for sufficiently large \(n\),
\[
\pp_0(\hat I_{n,p}(\mathcal C_p)>\hat D\, n^{-p/2}) \leq \pp_0(F_n^c)+\pp_0(E_{M}^c) < \epsilon,
\]
showing \(\hat I_{n,p}(\mathcal C_p)=O_{\pp_0}(n^{-p/2})\). Then, we conclude exactly as in the proof of Theorem~\ref{th:consistency-rates-h0}.
\end{proof}

\subsubsection{Proof of Theorem~\ref{th:consistency-rates-h1-plugin}}

\begin{proof}
We follow the proof strategy of Theorem~\ref{th:consistency-rates-h1}. Recover the approximating point \((p_\epsilon,\tau_\epsilon,b_\epsilon)\), the compact neighborhood \(U=U_1\times U_2\) of \((\tau_\epsilon, b_\epsilon)\) with positive prior probability, and the function \(f(\tau, b)=\mu_*(\tau)'b - \frac{1}{2}b'K_\tau b\), which satisfies \(f(\tau, b)\geq \frac{1}{2}\|\mu_*\|_K^2 - 2\epsilon^2 >0\) for all \((\tau,b)\in U\) and sufficiently small \(\epsilon>0\). Define the plug-in analogue
\[
\hat f_n(\tau,b) = \mu_*(\tau)'b - \frac{1}{2}b'\hat K_{n,\tau} b,\quad \text{for all } (\tau,b)\in U,
\]
and bound the plug-in error by
\[
\sup_{(\tau,b)\in U}|\hat f_n(\tau,b) - f(\tau,b)| = \sup_{(\tau,b)\in U}\frac{1}{2}|b'(\hat K_{n,\tau}-K_\tau)b| \leq \frac{B^2}{2}\sup_{\tau\in U_1}\|\hat K_{n,\tau}-K_\tau\|_2,
\]
where \(0<B=\sup_{b\in U_2} \|b\|<\infty\) and \(\|\cdot\|_2\) is the spectral norm on matrices. Now, define the event 
\[
G_n = \left\{\sup_{\tau\in U_1}\|\hat K_{n,\tau} - K_\tau\|_2 \leq \frac{2\epsilon^2}{B^2}\right\}.
\]
On the one hand, by~\ref{ass:uniform-consistency-plug-in} we have \(\pp_{\mu_*}(G_n^c)\to 0\) as \(n\to\infty\). On the other hand, on \(G_n\),
\begin{equation}\label{eq:plugin-f-lower-bound}
\hat f_n(\tau,b) \geq f(\tau,b) - \epsilon^2 \geq \frac{1}{2}\|\mu_*\|_K^2 - 3\epsilon^2 =: d_\epsilon, \quad \text{for all } (\tau,b)\in U,
\end{equation}
where \(d_\epsilon>0\) can be achieved by taking \(\epsilon\) small enough. In this case, the plug-in log-likelihood for \(p_\epsilon\) and \((\tau, b)\in U\) becomes
\begin{equation}\label{eq:log-likelihood-plugin}
\log \hat L(X_{1:n}|p_\epsilon, \tau, b) = n\mu_*(\tau)'b - \frac{n}{2}b'\hat K_{n,\tau}b + \sqrt{n}Z_\tau' b = n\hat f_n(\tau, b) + \sqrt{n}Z_\tau' b,
\end{equation}
where \(Z_\tau = \frac{1}{\sqrt{n}}(S_n(\tau) - n\mu_*(\tau))\). Moreover, recall that by Cauchy--Schwarz we have
\begin{equation}\label{eq:cauchy-schwarz-plugin}
\sqrt{n} Z_\tau' b \geq -\sqrt{n} B \sup_{\tau \in U_1}\|Z_\tau\|.
\end{equation}

On \(G_n\), restricting the domain of integration and using~\eqref{eq:plugin-f-lower-bound},~\eqref{eq:log-likelihood-plugin} and~\eqref{eq:cauchy-schwarz-plugin} produces the bound
\begin{equation}\label{eq:lower-bound-m1-plugin}
  \begin{split}
  \hat m_1(X_{1:n}) &\geq \Pi(p_\epsilon)\int_{U_1}\int_{U_2}\mathcal N_{p_\epsilon}(b|0,\eta^2 I_{p_\epsilon})\exp\left\{n\hat f_n(\tau,b)+\sqrt{n} Z_\tau'b\right\}\Pi(\tau|p_\epsilon)\,db\,d\tau \\
  &\geq D\exp\left\{nd_\epsilon - \sqrt{n}\,B\cdot\sup_{\tau\in U_1}\|Z_\tau\|\right\},
  \end{split}
\end{equation}
where \(D=\Pi(p_\epsilon)\Pi_\tau(U_1|p_{\epsilon})\Pi_b^{\text{slab}}(U_2|p_{\epsilon})>0\) by assumption. Now, since \(Z_\tau \sim \mathcal N_{p_\epsilon}(0, K_\tau)\) under \(H_1\), we already saw in the proof of Theorem~\ref{th:consistency-rates-h1} that for any \(\delta>0\) there exists \(M>0\) such that the event
\[
H_M = \left\{B\cdot \sup_{\tau\in U_1}\|Z_\tau\| \leq M\right\}
\]
satisfies \(\pp_{\mu_*}(H_M^c)<\delta/2\). Thus, we have the inequality \(\hat m_1(X_{1:n}) \geq D\exp\{nd_\epsilon - \sqrt{n}M\}\) on \(G_n\cap H_M\) by~\eqref{eq:lower-bound-m1-plugin}. Since \(nd_\epsilon/(1+\epsilon)\) grows slower than \(nd_\epsilon - \sqrt{n}M\) as \(n\to\infty\), for sufficiently large \(n\) we obtain \(\hat m_1(X_{1:n}) \geq D\exp\{n d_\epsilon/(1+\epsilon)\}\) on \(G_n\cap H_M\). Therefore,
\[
\hat\Pi_n(H_0|X_{1:n}) = \frac{\pi_0}{\pi_0+(1-\pi_0)\hat m_1(X_{1:n})}\leq \frac{\pi_0}{(1-\pi_0)\hat m_1(X_{1:n})}\leq \tilde{D}\exp\left\{\frac{-nd_\epsilon}{1+\epsilon}\right\}
\]
on \(G_n\cap H_M\), where \(\tilde{D}=\pi_0/((1-\pi_0)D)\). Write \(\pp_{\mu_*}(G_n^c\cup H_M^c)\leq \pp_{\mu_*}(G_n^c)+\pp_{\mu_*}(H_M^c)\), and note that the first term vanishes as \(n\to\infty\), while the second is at most \(\delta/2\). Then, for sufficiently large \(n\) we have proved that
\[
\pp_{\mu_*}\left(\hat\Pi_n(H_0|X_{1:n})>\tilde{D}\exp\{-nd_\epsilon/(1+\epsilon)\}\right)<\delta,
\]
so \(\hat\Pi_n(H_0|X_{1:n})=O_{\pp_{\mu_*}}(e^{-nd_\epsilon/(1+\epsilon)})\) and, as before, \(\lim_{\epsilon\to 0}\, d_\epsilon/(1+\epsilon)=\frac{1}{2}\|\mu_*\|_K^2\). The proof of the second part of the theorem is almost identical, except that \(f(\tau,b)\) is lower bounded by \(c_{\epsilon,\xi}\) as in the construction from the proof of Theorem~\ref{th:consistency-rates-h1-alt}.
\end{proof}

\subsection{On Optimality of the Contraction Rates}

We conclude with a proof of the information-theoretic optimality of the contraction rate in Theorem~\ref{th:consistency-rates-h1}. We will need the following lemma, which is just the standard fact that marginal likelihoods are densities with respect to the fixed reference measure. It follows directly from Fubini/Tonelli and the properties of Radon--Nikodym derivatives.

\begin{lemma}\label{lem:expectation-identity} For any \((p, \tau)\in \N \times \mathcal T_p\) and every \(n\geq 1\), we have \(\E_{\pp_0}[m_{\emph{slab}}(X_{1:n}|p,\tau)] = 1\).
\end{lemma}

\subsubsection{Proof of Proposition~\ref{prop:optimality-rates-h1}}

\begin{proof}
  Suppose by contradiction that a faster rate holds, say \(\Pi_n(H_0|X_{1:n})=O_{\pp_{\mu_*}}(e^{-n\zeta})\) for \(\zeta > \frac{1}{2}\|\mu_*\|_K^2\). By definition, for any \(0<\epsilon<1/2\) there exists \(M>0\) such that the region
  \[
  R_n = \{X_{1:n}: \Pi_n(H_0|X_{1:n}) \leq M e^{-n\zeta}\}
  \]
  satisfies \(\pp_{\mu_*}(R_n)> 1-\epsilon\) for sufficiently large \(n\). Then, reusing the notation \(m_1(X_{1:n})\) from~\eqref{eq:posterior-slab-m1}, on \(R_n\) we have 
  \[
  \Pi_n(H_0|X_{1:n}) = \frac{\pi_0}{\pi_0 + (1-\pi_0)m_1(X_{1:n})} \leq M e^{-n\zeta},
  \]
  and rearranging yields
  \[
  \pi_0(1-M e^{-n\zeta}) \leq (1-\pi_0)m_1(X_{1:n}) M e^{-n\zeta}.
  \]
  For sufficiently large \(n\) we have \(1-M e^{-n\zeta} \geq 1/2\), so eventually we get
  \[
  1\leq \frac{2M(1-\pi_0)}{\pi_0} e^{-n\zeta} m_1(X_{1:n}).
  \]
  Integrating on \(R_n\) on both sides with respect to \(\pp_0\) and using the fact that \(m_1(X_{1:n})\) is a valid marginal likelihood (it integrates to 1 under \(\pp_0\) over the whole space by Lemma~\ref{lem:expectation-identity}), we have
  \begin{equation}\label{eq:type-ii-error-bound}
  \pp_0(R_n) \leq \frac{2M(1-\pi_0)}{\pi_0} e^{-n\zeta} \int_{R_n} m_1(X_{1:n})\, d\pp_0 \leq C e^{-n\zeta}\quad \text{for large } n,
  \end{equation}
  where \(C>0\) is a constant. Now we construct a frequentist test \(\phi_n\) for the reverse hypotheses \(\tilde{H}_0: \mu\equiv \mu_*\) against \(\tilde{H}_1: \mu\equiv 0\), given by \(\phi_n(X_{1:n}) = \mathds{1}\{X_{1:n} \in R_n^c\}\). In this case, according to~\eqref{eq:rn-derivative} the Kullback--Leibler divergence from \(\pp_{\mu_*}\) to \(\pp_0\) is
  \[
  \kl(\pp_{\mu_*} \,\|\, \pp_0) = \E_{\pp_{\mu_*}}\left[\log \frac{d\pp_{\mu_*}}{d\pp_0}(X)\right] = \E_{\pp_{\mu_*}}[\dotprod{X}{\mu_*}_K] - \frac{1}{2}\|\mu_*\|^2_K = \frac{1}{2}\|\mu_*\|_K^2,
  \]
  where we used the equivalence \(\E_{\pp_{\mu_*}}[\dotprod{X}{\mu_*}_K]=\dotprod{\mu_*}{\mu_*}_K=\|\mu_*\|_K^2\) that stems from Loève's isometry \citep[p.~974]{parzen1961approach}. For sufficiently large \(n\), the type I error probability is bounded, since \(\pp_{\mu_*}(\phi_n = 1)=\pp_{\mu_*}(R_n^c) < \epsilon\). On the other hand, by~\eqref{eq:type-ii-error-bound} the type II error probability for large \(n\) is \(\pp_0(\phi_n=0)=\pp_0(R_n)\leq C e^{-n\zeta}\), which decays exponentially with exponential constant greater than \(\kl(\pp_{\mu_*} \,\|\, \pp_0)\), contradicting Chernoff--Stein.
\end{proof}

\section{Additional Information on Numerical Experiments}\label{app:experiments}

We provide a description of the comparison methods, a runtime comparison of our RKBT test, and some technical implementation details.

\subsection{Description of Comparison Methods}

The frequentist comparison methods were mostly described by~\citet{zhang2013analysis}. We consider the pivotal test function \(\Delta_n(s) = \sqrt{n}\bar X_n(s)\), which will be small when the null hypothesis is true. We also fix a significance level \(\alpha\in (0,1)\) and the discretized grid \(\mathcal S=\{s_1,\dots,s_g\}\).

\subsubsection{Pointwise \(t\)-Test}

The first naive idea is to construct a pointwise \(t\)-test over all points on the grid, and then reject \(H_0\) if any pointwise test rejects \(H_{0s}: \mu(s)=0\). The local test statistics are
\[
z_n(s) = \frac{\Delta_n(s)}{\sqrt{\hat C_n(s,s)}} \sim t_{n-1},\quad s\in[0,1].
\]

We use the Bonferroni correction on the \(p\)-values~\citep{shaffer1995multiple} to produce a global test that accounts for the multiple comparisons problem and controls the probability of the so-called family-wise error rate (FWER), at the expense of decreased power. In particular, we reject \(H_{0}\) whenever the \(p\)-value of at least one of the local tests is less than \(\alpha/G\). While there are more sophisticated approaches to correct \(p\)-values \citep[e.g.][]{cox2008pointwise}, we keep the simple Bonferroni method as a baseline.

\subsubsection{\(L^2\)-Norm Based Tests}

We perform a global test based on the squared \(L^2\)-norm \(\|\cdot\|_2\) of the pivotal test function \(\Delta_n\). The proposed test statistic is
\begin{equation}\label{eq:test-statistic-l2}
T_n = \|\Delta_n\|^2_2 = n\int_0^1 \bar{X}_n^2(s)\, ds.
\end{equation}
Under \(H_0\) and Gaussian data, \(T_n\) follows a \(\chi^2\)-type mixture distribution. Using a Box-type approximation \citep{box1954some}, we have
\[
T_n \sim \gamma \chi^2_d\,  \text{ approximately}, \quad \text{where } \gamma = \frac{\tr(K^{\otimes 2})}{\tr(K)}\, \text{ and }\, d=\frac{\tr^2(K)}{\tr(K^{\otimes 2})},
\]
with \(\tr(K^{\otimes 2}) = \int_0^1 \int_0^1 K(t,s)^2 \, dt\, ds\). Since \(K\) is generally unknown, a common approach is to use the empirical covariance estimator \(\hat C_n\) to approximate \(\gamma\) and \(d\), which leads to consistent estimators \(\hat\gamma\) and \(\hat d\) \citep[Theorem~4.17]{zhang2013analysis}. The corresponding \(p\)-value for the test is given by \(\pp(\chi^2_{\hat d} \geq T_n/\hat \gamma)\). In practice, we need to discretize the functional elements to compute the test statistic and \(p\)-value:
\[
T_n = n\sum_{g=1}^G \bar{X}_n^2(s_g), \qquad \hat \gamma \approx \frac{\tr(\hat{\mathbf C}_n^2)}{\tr(\hat{\mathbf C}_n)}, \qquad \hat d \approx \frac{\tr^2(\hat{\mathbf C}_n)}{\tr(\hat{\mathbf C}_n^2)}.
\]
Note that the constant discretization step \(1/G\) that would appear from approximating the integrals cancels out in the computations, so it can be omitted for simplicity.

\begin{remark}
  If we use the improved estimator \(\hat{\mathbf K}_n^{\text{OAS}}\) to approximate \(\hat \gamma\) and \(\hat d\), the argument goes through unaffected: by Proposition~\ref{prop:sample-cov-matrix-convergence} we have \(\|\hat{\mathbf K}^{\text{OAS}}_n - \mathbf{K}\| \to 0\) a.s., and hence the traces converge almost surely to the true value.
\end{remark}

We can also consider a projection-based version of the \(L^2\) test using functional principal component analysis (FPCA), which was explored for example by \citet{horvath2012inference} in a similar testing problem. Let \(\lambda_1\geq \lambda_2\geq \cdots \geq \lambda_p>0\) be the \(p\) leading eigenvalues of the true covariance operator, with \(\phi_1,\dots,\phi_p\) the corresponding eigenfunctions (up to a sign flip). Then, we can define a test statistic analogous to~\eqref{eq:test-statistic-l2} but with the mean function projected on the space spanned by \(\phi_1,\dots,\phi_p\). Write the FPCA scores of the sample mean as
\begin{equation}\label{eq:l2-fpca-scores}
\bar{\xi}_j = \int_0^1 \bar X_n(s)\phi_j(s)\, ds,\quad j=1,\dots,p,
\end{equation}
and consider the statistic
\begin{equation}\label{eq:statistic-l2-fpca}
T_{n,p} = n \sum_{j=1}^p \frac{\bar{\xi}_j^2}{\lambda_j},
\end{equation}
where normalization by the eigenvalues is introduced to simplify the expression of the null distribution. Indeed, by standard FPCA theory, under \(H_0\) and the Gaussian assumption we have \(T_{n,p}\sim \chi^2_p\). Of course, we rarely have access to the true \(K\), so we construct estimates \(\{\hat\lambda_j\}\) and \(\{\hat\phi_j\}\) for the eigenvalues and eigenfunctions via the (truncated) spectral decomposition of the empirical covariance \(\hat C_n\), and then compute an approximate \(\hat T_{n,p}\) plugging in the estimated values into~\eqref{eq:l2-fpca-scores} and~\eqref{eq:statistic-l2-fpca}. Under mild assumptions (finite fourth moments, distinctness of the true eigenvalues), it can be shown that \(\hat T_{n,p}\stackrel{d}{\to}\chi^2_p\) as \(n\to\infty\) \citep[Theorem~5.3]{horvath2012inference}, so asymptotically the \(p\)-value for the test is \(\pp(\chi^2_p \geq \hat T_{n,p})\). As we said before, in practice we compute all quantities involved in the discretized grid \(\mathcal S\):
\[
\hat \lambda_j \approx \hat l_j,\qquad \hat\phi_j \approx \hat{\mathbf v}_j=(\hat v_{j1},\dots, \hat v_{jG})', \qquad \hat T_{n,p} \approx n\sum_{j=1}^p \frac{(\sum_{g=1}^G \bar{X}_n(s_g)\hat v_{jg})^2}{\hat l_j}, 
\]
where \(\{\hat l_j, \hat{\mathbf v}_j\}\) are the eigenvalues and (normalized) eigenvectors of the \(G\times G\) empirical covariance matrix \(\hat{\mathbf C}_n\). We could also use the regularized version \(\hat{\mathbf K}_n^{\text{OAS}}\) to compute the eigendecomposition, and even though we did not explore the theoretical guarantees in this case, it is likely that the eigenvalues and eigenvectors of the OAS estimator constitute sufficiently good approximations for our purposes.

The only missing aspect is the question of how many components to retain. We follow the simple approach proposed by \citet{ramsay2005functional}, who suggest selecting \(p\) so that the resulting functional principal components cumulatively explain a fixed large percentage of the variance (e.g.\ 90\%).

\subsubsection{\(F\)-Type Test}

The previous \(L^2\)-based tests do not take into account the variation of the covariance function. Under \(H_0\) and Gaussian data, we have \(\E\|\Delta_n\|^2_2 = \tr(K)\) and \(\E[\tr(\hat{C}_n)] = \tr(K)\), so we can consider the statistic
\[
F_n = \frac{\|\Delta_n\|^2_2}{\tr(\hat C_n)} \approx \frac{n\sum_{g=1}^G \bar{X}_n^2(s_g)}{\tr(\hat{\mathbf C}_n)},
\]
whose null distribution is shown to be an \(F\)-type mixture. It can again be approximated by a two-cumulant method \citep[\S4.4]{zhang2013analysis}, so that \(F_n \sim F_{\hat d, (n-1)\hat d}\) approximately, and the approximate \(p\)-value for the test is \(\pp(F_{\hat d, (n-1)\hat d} \geq F_n)\). This method works for the OAS estimator as well, since \(\tr(\hat{\mathbf K}_n^{\text{OAS}})=\tr(\hat{\mathbf K}_n)\).

\subsubsection{Bootstrap Procedures}

Both the \(L^2\)-based tests and the \(F\)-type test rely on asymptotic approximations, so their performance can deteriorate when the sample size \(n\) is small. For such cases we can resort to a nonparametric bootstrap scheme to approximate the null distribution of the statistics. In particular, we consider the following simple algorithm:
\begin{enumerate}
  \item Draw a bootstrap sample \(X_1^*,\dots, X_n^*\) from \(X_1,\dots, X_n\) (with replacement).
  \item Compute the quantities \(\bar{X}_n^*\) and \(\hat{\mathbf C}_n^*\) from the bootstrap sample.
  \item Construct the bootstrap statistics \(T_n^* = \|\sqrt{n}(\bar{X}_n^*(s) - \bar{X}(s))\|^2_2\) and \(F_n^* = T_n^*/\tr(\hat{\mathbf C}_n^*)\).
  \item Repeat the process many times to obtain \(R\) bootstrap statistics, and compute the bootstrap \(p\)-value as the proportion of bootstrap statistics that exceed the original statistic.
\end{enumerate}
In particular, the bootstrap \(p\)-values are
\[
p^*(T_n) = \frac{1 + \#\{T_n^* \geq T_n\}}{R+1},\qquad p^*(F_n) = \frac{1 + \#\{F_n^* \geq F_n\}}{R+1},
\]
where the original statistic is added as a positive instance to avoid the case of \(p^*=0\) \citep[see][p.~141]{davison1997bootstrap}. Note that the same procedure can be applied to the \(L^2\)-FPCA test, recomputing the eigendecomposition of the bootstrap empirical covariance matrix \(\hat{\mathbf C}_n^*\) each time. In our implementation, we activate this bootstrap procedure with \(R=5000\) whenever \(n\leq 50\).

\subsubsection{Bayesian Functional Test with FPCA}

Consider for all \(i=1,\dots, n\) the discretized vector \(\mathbf X_i=(X_i(s_1), \dots, X_i(s_G))'\), and place a multivariate normal model on these finite-dimensional observations, that is,
\[
\mathbf X_i|\boldsymbol{\mu}, \mathbf K \sim \mathcal N_G(\boldsymbol \mu, \mathbf K),\quad  i=1,\dots,n,
\]
where \(\boldsymbol\mu=(\mu(s_1), \dots, \mu(s_G))'\) and \(\mathbf K\) is the \(G\times G\) discretized covariance on the grid. To test \(H_0: \boldsymbol\mu = 0\) against \(H_1: \boldsymbol\mu \neq 0\), we place a \(g\)-prior on the mean under \(H_1\), i.e.,
\[
\boldsymbol \mu | g, \mathbf K \sim \mathcal{N}_G(0, g \mathbf K),
\]
where \(g>0\) is a scalar hyperparameter. This prior, introduced by \citet{zellner1986assessing}, encodes the belief that the mean function varies along the same directions as the data covariance, scaled by the factor \(g\). Under this conjugate model, the Bayes factor in favor of the alternative is available in closed form after integrating out \(\boldsymbol \mu\) \citep[for a similar computation, see][]{liang2008mixtures}:
\begin{equation}\label{eq:log-bf10}
\log \bff_{10} = -\frac{G}{2}\log(1+ng) + \frac{n^2 g}{2(1+ng)}Q_n,
\end{equation}
where \(Q_n = \bar{\mathbf X}_n' \mathbf K^{-1}\bar{\mathbf X}_n\), \(\bar{\mathbf X}_n = \frac{1}{n}\sum_{i=1}^n \mathbf X_i\) is the discretized sample mean, and we have supposed for simplicity that \(\pi_0=1-\pi_0=1/2\).

In practice, we compute the Bayes factor using the plug-in estimators \(\hat{\mathbf C}_n\) or \(\hat {\mathbf K}_n^{\text{OAS}}\) of the covariance matrix. However, applying~\eqref{eq:log-bf10} directly on the full \(G\)-dimensional discretized data presents two noticeable difficulties:
\begin{enumerate}
  \item When \(n<G\) (the typical FDA setting), the empirical covariance matrix \(\hat{\mathbf C}_n\) is rank-deficient, so \(\hat{\mathbf C}_n^{-1}\) and the corresponding plug-in quadratic form \(\hat Q_n\) do not exist. Even with regularization (e.g.\ by OAS or adding a small noise to the diagonal), the resulting \(\hat Q_n\) can be numerically unstable and a poor estimate.
  \item The Bayes factor contains a penalty term that grows linearly with \(G\). This penalty can become large and drown out the signal in \(Q_n\), driving \(\bff_{10}\) towards \(H_0\) regardless of the data. This phenomenon is known as Lindley's paradox \citep{lindley1957statistical}.
\end{enumerate}

From a functional perspective, both problems are organically resolved by projecting onto the top \(p\ll G\) functional principal components, retaining the signal-carrying directions and discarding the noisy ones. As before, consider the eigendecomposition
\[
\mathbf K = \mathbf V \Sigma \mathbf V', \qquad \mathbf V = [\mathbf v_1 | \cdots | \mathbf v_p], \qquad \Sigma = \diag(l_1,\dots, l_p),
\]
where \(l_1 \geq \cdots \geq l_p > 0\) are the leading eigenvalues of \(\mathbf K\) and \(\mathbf v_1,\dots, \mathbf v_p\) are the corresponding eigenvectors. We can transform the data onto the \(p\)-dimensional score space, defining the FPCA scores as
\[
\xi_i = \mathbf V' \mathbf X_i, \quad i=1,\dots,n.
\]
In score space, the null hypothesis can be expressed as \(H_0: \tilde{\mu}=0\), where \(\tilde{\mu}=(\tilde \mu_1,\dots, \tilde \mu_p)'\) and \(\tilde \mu_j = \boldsymbol \mu' \mathbf v_j\) is the projection of the true mean vector onto the \(j\)-th eigenvector. In this way, the model becomes
\[
\xi_i|\tilde{\mu}, \Sigma \sim \mathcal N_p(\tilde{\mu}, \Sigma), \quad i=1,\dots,n,
\]
and the corresponding \(g\)-prior under \(H_1\) is
\[
\tilde{\mu}|g, \Sigma \sim \mathcal N_p(0, g\Sigma).
\]
We can derive the Bayes factor for this reduced model in the same way as we did above for the fully discretized model (in fact now it is easier because the prior factorizes in independent components), obtaining
\begin{equation}\label{eq:log-bf10-fpca}
\log \bff_{10}^{\text{FPCA}} = -\frac{p}{2}\log(1+ng) + \frac{n^2 g}{2(1+ng)} Q_{n,p},
\end{equation}
where 
\[
Q_{n,p} = \sum_{j=1}^p \frac{(\bar {\mathbf X}_n' \mathbf v_j)^2}{l_j}.
\]
Observe that \(Q_{n,p}\) is essentially a discretized and rescaled version of the statistic \(T_{n,p}\) in the \(L^2\)-based FPCA test, defined in~\eqref{eq:statistic-l2-fpca}. Naturally, in the actual implementation we approximate the eigenvalues and eigenvectors via the eigendecomposition of the empirical covariance matrix \(\hat {\mathbf C}_n\) or its regularized version, obtaining empirical estimates \(\{\hat l_j, \hat{\mathbf v}_j\}\) that we can substitute into the formulas above.

We still need to choose a suitable value for the hyperparameter \(g\). It controls the prior spread under \(H_1\): a small \(g\) concentrates the prior near the null, while a large \(g\) makes it diffuse. A fixed choice (e.g.\ \(g=1\)) is common, but data-driven selection typically performs better. As mentioned by~\citet{liang2008mixtures}, we can compute an empirical Bayes estimate of \(g\) by maximizing the marginal likelihood under \(H_1\). In particular, consider
\[
g^*_n = \operatorname*{argmax}_g\, \bff_{10}^{\text{FPCA}}(g) = \frac{Q_{n,p}}{p} - \frac{1}{n},
\]
which is obtained by differentiating~\eqref{eq:log-bf10-fpca} with respect to \(g\) and setting the derivative to zero. Then, the estimated value for \(g\) is \(\hat g_n = \max\{\delta, g^*_n\}\), where \(\delta>0\) is a small constant to avoid degeneracy at the boundary \(\{g=0\}\) (in the experiments we used \(\delta=10^{-3}\)).

Note that this FPCA-based Bayesian method is simpler than the full RJMCMC procedure described in Section~\ref{sec:rkhs-test}, and provides an exact Bayes factor without the need to approximate the posterior. However, it is not as flexible as the RKHS-based approach, and it seems to perform a bit worse in practice.

\subsection{Practical Notes for Implementation}

We comment briefly on a few implementation suggestions aimed at improving the computational efficiency of the tests:

\begin{enumerate}
  \item Use Cholesky decompositions to avoid computing expensive matrix inverses.
  \item Use log-formulas for quantities like \(\mslab(X_{1:n}|p,\tau)\), and employ the built-in Python functions \texttt{logaddexp(x,y)=log(exp(x)+exp(y))} and \texttt{log1p(x)=log(1+x)} to compute the log-marginal likelihood, the Bayes factor and other expressions.
  \item Store values such as \(\log \mslab(X_{1:n}|p^{(r)}, \tau^{(r)})\) or \((\Lambda_\tau^{-1})^{(r)}\) per RJMCMC step to reuse them in later computations.
  \item If needed, use a stable sigmoid to compute the posterior probability from the log-Bayes factor: if \(\ell=\log(\frac{1-\pi_0}{\pi_0}\bff_{10})\geq 0\), then \(\Pi(H_1|X) = 1/(1+e^{-\ell})\); otherwise, \(\Pi(H_1|X)=e^\ell/(1+e^\ell)\).
\end{enumerate}

\subsection{Runtime Comparison}

As we mentioned before, our test is slower than the competitors due to the RJMCMC sampling, but it still runs in acceptable time under the experimental conditions in Section~\ref{sec:experiments} (see Figure~\ref{fig:runtime}). The \(L^2\)-FPCA method is also slow because it requires recomputing the FPCA decomposition for each bootstrap sample. The rest of the methods are almost instantaneous, as they do not perform intensive computation.

\begin{figure}[t]
  \centering
  \includegraphics[width=.9\textwidth]{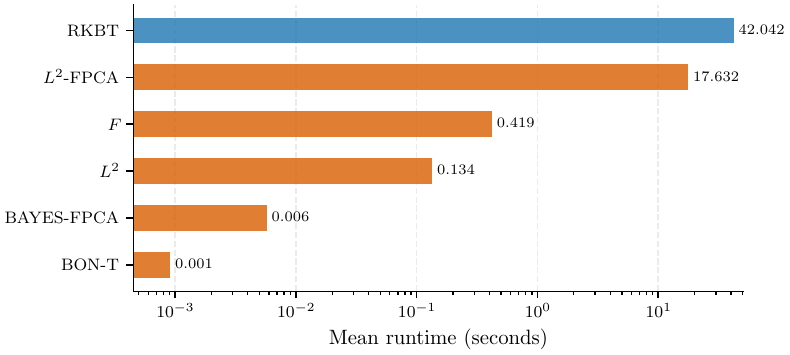}
  \caption{Mean runtime of all tests for the synthetic data sets (log-scale x-axis).}\label{fig:runtime}
\end{figure}

\section{Experimental Setup}

All experiments were run on a Linux laptop with an AMD Ryzen 7 6800HS CPU and 16 GB of RAM. The Python code for reproducing the experiments is publicly available online at \url{https://github.com/antcc/rk-bt-jump}.

\bibliography{bibliography}

@article{berger1987testing,
  title   = {Testing precise hypotheses},
  author  = {Berger, James O and Delampady, Mohan},
  year    = {1987},
  journal = {Statistical Science},
  volume  = {2},
  number  = {3},
  pages   = {317--335},
  doi     = {10.1214/ss/1177013238},
}

@article{berkes2009detecting,
  title   = {Detecting changes in the mean of functional observations},
  author  = {Berkes, Istv{\'a}n and Gabrys, Robertas and Horv{\'a}th, Lajos and Kokoszka, Piotr},
  year    = {2009},
  journal = {Journal of the Royal Statistical Society Series B: Statistical Methodology},
  volume  = {71},
  number  = {5},
  pages   = {927--946},
  doi     = {10.1111/j.1467-9868.2009.00713.x},
}

@book{berlinet2004reproducing,
  title     = {Reproducing Kernel {Hilbert} Spaces in Probability and Statistics},
  author    = {Berlinet, Alain and Thomas-Agnan, Christine},
  year      = {2004},
  doi       = {10.1007/978-1-4419-9096-9},
  publisher = {Springer},
}

@article{berrendero2025bayesian,
  title   = {A {Bayesian} approach to functional regression: theory and computation},
  author  = {Berrendero, Jos{\'e} R and Co{\'i}n, Antonio and Cuevas, Antonio},
  year    = {2025},
  journal = {Bayesian Analysis},
  volume  = {Advance Publication},
  pages   = {1--31},
  doi     = {10.1214/25-BA1549},
}

@article{berrendero2018use,
  title   = {On the use of reproducing kernel {Hilbert} spaces in functional classification},
  author  = {Berrendero, Jos{\'e} R and Cuevas, Antonio and Torrecilla, Jos{\'e} L},
  year    = {2018},
  journal = {Journal of the American Statistical Association},
  volume  = {113},
  number  = {523},
  pages   = {1210--1218},
  doi     = {10.1080/01621459.2017.1320287},
}

@book{bhatia1996matrix,
  title     = {Matrix Analysis},
  author    = {Bhatia, Rajendra},
  year      = {1996},
  doi       = {10.1007/978-1-4612-0653-8},
  publisher = {Springer},
}

@article{box1954some,
  title   = {Some theorems on quadratic forms applied in the study of analysis of variance problems, {I}. {Effect} of inequality of variance in the one-way classification},
  author  = {Box, G E P},
  year    = {1954},
  journal = {The Annals of Mathematical Statistics},
  volume  = {25},
  number  = {2},
  pages   = {290--302},
  doi     = {10.1214/aoms/1177728786},
}

@book{brooks2011handbook,
  title     = {Handbook of {Markov} Chain {Monte} {Carlo}},
  author    = {Brooks, Steve and Gelman, Andrew and Jones, Galin and Meng, Xiao-Li},
  year      = {2011},
  doi       = {10.1201/b10905},
  publisher = {Chapman and Hall/CRC},
}

@article{cardot2003testing,
  title   = {Testing hypotheses in the functional linear model},
  author  = {Cardot, Herv{\'e} and Ferraty, Fr{\'e}d{\'e}ric and Mas, Andr{\'e} and Sarda, Pascal},
  year    = {2003},
  journal = {Scandinavian Journal of Statistics},
  volume  = {30},
  number  = {1},
  pages   = {241--255},
  doi     = {10.1111/1467-9469.00329},
}

@article{chen2010shrinkage,
  title   = {Shrinkage algorithms for {MMSE} covariance estimation},
  author  = {Chen, Yilun and Wiesel, Ami and Eldar, Yonina C and Hero, Alfred O},
  year    = {2010},
  journal = {IEEE Transactions on Signal Processing},
  volume  = {58},
  number  = {10},
  pages   = {5016--5029},
  doi     = {10.1109/TSP.2010.2053029},
}

@book{cover2006elements,
  title     = {Elements of Information Theory},
  author    = {Cover, Thomas M and Thomas, Joy A},
  year      = {2006},
  doi       = {10.1002/047174882X},
  publisher = {John Wiley \& Sons},
}

@article{cox2008pointwise,
  title   = {Pointwise testing with functional data using the {Westfall--Young} randomization method},
  author  = {Cox, Dennis D and Lee, Jong Soo},
  year    = {2008},
  journal = {Biometrika},
  volume  = {95},
  number  = {3},
  pages   = {621--634},
  doi     = {10.1093/biomet/asn021},
}

@article{crainiceanu2007spatially,
  title   = {Spatially adaptive {Bayesian} penalized splines with heteroscedastic errors},
  author  = {Crainiceanu, Ciprian M and Ruppert, David and Carroll, Raymond J and Joshi, Adarsh and Goodner, Billy},
  year    = {2007},
  journal = {Journal of Computational and Graphical Statistics},
  volume  = {16},
  number  = {2},
  pages   = {265--288},
  doi     = {10.1198/106186007X208768},
}

@article{cuesta2010simple,
  title   = {A simple multiway {ANOVA} for functional data},
  author  = {Cuesta-Albertos, J. A. and Febrero-Bande, Manuel},
  year    = {2010},
  journal = {{TEST}},
  volume  = {19},
  number  = {3},
  pages   = {537--557},
  doi     = {10.1007/s11749-010-0185-3},
}

@article{cuevas2004anova,
  title   = {An {ANOVA} test for functional data},
  author  = {Cuevas, Antonio and Febrero-Bande, Manuel and Fraiman, Ricardo},
  year    = {2004},
  journal = {Computational Statistics \& Data Analysis},
  volume  = {47},
  number  = {1},
  pages   = {111--122},
  doi     = {10.1016/j.csda.2003.10.021},
}

@article{dass2004note,
  title   = {A note on the consistency of {Bayes} factors for testing point null versus non-parametric alternatives},
  author  = {Dass, Sarat C and Lee, Jaeyong},
  year    = {2004},
  journal = {Journal of Statistical Planning and Inference},
  volume  = {119},
  number  = {1},
  pages   = {143--152},
  doi     = {10.1016/S0378-3758(02)00413-5},
}

@book{davison1997bootstrap,
  title     = {Bootstrap Methods and Their Application},
  author    = {Davison, A. C. and Hinkley, D. V.},
  year      = {1997},
  doi       = {10.1017/CBO9780511802843},
  publisher = {Cambridge University Press},
}

@article{delaigle2012achieving,
  title   = {Achieving near perfect classification for functional data},
  author  = {Delaigle, Aurore and Hall, Peter},
  year    = {2012},
  journal = {Journal of the Royal Statistical Society Series B: Statistical Methodology},
  volume  = {74},
  number  = {2},
  pages   = {267--286},
  doi     = {10.1111/j.1467-9868.2011.01003.x},
}

@article{fan1998test,
  title   = {Test of significance when data are curves},
  author  = {Fan, Jianqing and Lin, Sheng-Kuei},
  year    = {1998},
  journal = {Journal of the American Statistical Association},
  volume  = {93},
  number  = {443},
  pages   = {1007--1021},
  doi     = {10.1080/01621459.1998.10473763},
}

@article{feldman1958equivalence,
  title   = {Equivalence and perpendicularity of {Gaussian} processes},
  author  = {Feldman, Jacob},
  year    = {1958},
  journal = {Pacific Journal of Mathematics},
  volume  = {8},
  number  = {4},
  pages   = {699--708},
}

@book{ghosal2017fundamentals,
  title     = {Fundamentals of Nonparametric {Bayesian} Inference},
  author    = {Ghosal, Subhashis and van der Vaart, Aad W},
  year      = {2017},
  doi       = {10.1017/9781139029834},
  publisher = {Cambridge University Press},
}

@book{ghosh2006introduction,
  title     = {An Introduction to {Bayesian} Analysis: Theory and Methods},
  author    = {Ghosh, Jayanta K and Delampady, Mohan and Samanta, Tapas},
  year      = {2006},
  doi       = {10.1007/978-0-387-35433-0},
  publisher = {Springer},
}

@article{green1995reversible,
  title   = {Reversible jump {Markov} chain {Monte} {Carlo} computation and {Bayesian} model determination},
  author  = {Green, Peter J},
  year    = {1995},
  journal = {Biometrika},
  volume  = {82},
  number  = {4},
  pages   = {711--732},
  doi     = {10.2307/2337340},
}

@article{hormann2010weakly,
  title   = {Weakly dependent functional data},
  author  = {H{\"o}rmann, Siegfried and Kokoszka, Piotr},
  year    = {2010},
  journal = {The Annals of Statistics},
  volume  = {38},
  number  = {3},
  pages   = {1845--1884},
  doi     = {10.1214/09-AOS768},
}

@article{hall2007two,
  title   = {Two-sample tests in functional data analysis starting from discrete data},
  author  = {Hall, Peter and Van Keilegom, Ingrid},
  year    = {2007},
  journal = {Statistica Sinica},
  volume  = {17},
  number  = {4},
  pages   = {1511--1531},
  url     = {https://www3.stat.sinica.edu.tw/statistica/j17n4/j17n412/j17n412.html},
}

@book{horn1985matrix,
  title     = {Matrix Analysis},
  author    = {Horn, Roger A and Johnson, Charles R},
  year      = {1985},
  doi       = {10.1017/CBO9780511810817},
  publisher = {Cambridge University Press},
}

@book{horvath2012inference,
  title     = {Inference for Functional Data with Applications},
  author    = {Horv{\'a}th, Lajos and Kokoszka, Piotr},
  year      = {2012},
  doi       = {10.1007/978-1-4614-3655-3},
  publisher = {Springer},
}

@book{hsing2015theoretical,
  title     = {Theoretical Foundations of Functional Data Analysis, with an Introduction to Linear Operators},
  author    = {Hsing, Tailen and Eubank, Randall},
  year      = {2015},
  doi       = {10.1002/9781118762547},
  publisher = {John Wiley \& Sons},
}

@article{ishwaran2005spike,
  title   = {Spike and slab variable selection: frequentist and {Bayesian} strategies},
  author  = {Ishwaran, Hemant and Rao, J Sunil},
  year    = {2005},
  journal = {The Annals of Statistics},
  volume  = {33},
  number  = {2},
  pages   = {730--773},
  doi     = {10.1214/009053604000001147},
}

@book{jeffreys1961theory,
  title     = {Theory of Probability},
  author    = {Jeffreys, Harold},
  year      = {1961},
  edition   = {3rd},
  publisher = {Oxford University Press},
}

@article{johnson2010use,
  title   = {On the use of non-local prior densities in {Bayesian} hypothesis tests},
  author  = {Johnson, Valen E and Rossell, David},
  year    = {2010},
  journal = {Journal of the Royal Statistical Society Series B: Statistical Methodology},
  volume  = {72},
  number  = {2},
  pages   = {143--170},
  doi     = {10.1111/j.1467-9868.2009.00730.x},
}

@article{karnesis2023eryn,
  title   = {{Eryn}: a multipurpose sampler for {Bayesian} inference},
  author  = {Karnesis, Nikolaos and Katz, Michael L and Korsakova, Natalia and Gair, Jonathan R and Stergioulas, Nikolaos},
  year    = {2023},
  journal = {Monthly Notices of the Royal Astronomical Society},
  volume  = {526},
  number  = {4},
  pages   = {4814--4830},
  doi     = {10.1093/mnras/stad2939},
}

@article{kass1995bayes,
  title   = {{Bayes} factors},
  author  = {Kass, Robert E and Raftery, Adrian E},
  year    = {1995},
  journal = {Journal of the American Statistical Association},
  volume  = {90},
  number  = {430},
  pages   = {773--795},
  doi     = {10.1080/01621459.1995.10476572},
}

@article{kokoszka2008testing,
  title   = {Testing for lack of dependence in the functional linear model},
  author  = {Kokoszka, Piotr and Maslova, Inga and Sojka, Jan and Zhu, Lie},
  year    = {2008},
  journal = {Canadian Journal of Statistics},
  volume  = {36},
  number  = {2},
  pages   = {207--222},
  doi     = {10.1002/cjs.5550360203},
}

@book{krzanowski2009roc,
  title     = {{ROC} Curves for Continuous Data},
  author    = {Krzanowski, Wojtek J and Hand, David J},
  year      = {2009},
  doi       = {10.1201/9781439800225},
  publisher = {Chapman and Hall/CRC},
}

@book{laha1979probability,
  title     = {Probability Theory},
  author    = {Laha, R. G. and Rohatgi, V. K.},
  year      = {1979},
  publisher = {John Wiley \& Sons},
}

@misc{lamprecht2026meteostat,
  title  = {{Meteostat} {Python} (version 2.1.4)},
  author = {Lamprecht, Christian Sebastian},
  year   = {2026},
  url    = {https://github.com/meteostat/meteostat},
}

@article{landau1970supremum,
  title   = {On the supremum of a {Gaussian} process},
  author  = {Landau, Henry J and Shepp, Lawrence A},
  year    = {1970},
  journal = {Sankhy{\=a}: The Indian Journal of Statistics, Series A},
  volume  = {32},
  number  = {4},
  pages   = {369--378},
}

@article{ledoit2004well,
  title   = {A well-conditioned estimator for large-dimensional covariance matrices},
  author  = {Ledoit, Olivier and Wolf, Michael},
  year    = {2004},
  journal = {Journal of Multivariate Analysis},
  volume  = {88},
  number  = {2},
  pages   = {365--411},
  doi     = {10.1016/S0047-259X(03)00096-4},
}

@article{li2010uniform,
  title   = {Uniform convergence rates for nonparametric regression and principal component analysis in functional/longitudinal data},
  author  = {Li, Yehua and Hsing, Tailen},
  year    = {2010},
  journal = {The Annals of Statistics},
  volume  = {38},
  number  = {6},
  pages   = {3321--3351},
  doi     = {10.1214/10-AOS813},
}

@article{liang2008mixtures,
  title   = {Mixtures of $g$ priors for {Bayesian} variable selection},
  author  = {Liang, Feng and Paulo, Rui and Molina, German and Clyde, Merlise A and Berger, James O},
  year    = {2008},
  journal = {Journal of the American Statistical Association},
  volume  = {103},
  number  = {481},
  pages   = {410--423},
  doi     = {10.1198/016214507000001337},
}

@article{lindley1957statistical,
  title   = {A statistical paradox},
  author  = {Lindley, Dennis V},
  year    = {1957},
  journal = {Biometrika},
  volume  = {44},
  number  = {1--2},
  pages   = {187--192},
  doi     = {10.1093/biomet/44.1-2.187},
}

@article{lukic2001stochastic,
  title   = {Stochastic processes with sample paths in reproducing kernel {Hilbert} spaces},
  author  = {Luki{\'c}, Milan N and Beder, Jay H},
  year    = {2001},
  journal = {Transactions of the American Mathematical Society},
  volume  = {353},
  number  = {10},
  pages   = {3945--3969},
  doi     = {10.1090/s0002-9947-01-02852-5},
}

@article{martinez2011repeated,
  title   = {Repeated measures analysis for functional data},
  author  = {Mart{\'\i}nez-Camblor, Pablo and Corral, Norberto},
  year    = {2011},
  journal = {Computational Statistics \& Data Analysis},
  volume  = {55},
  number  = {12},
  pages   = {3244--3256},
  doi     = {10.1016/j.csda.2011.06.007},
}

@article{mas2007testing,
  title   = {Testing for the mean of random curves: a penalization approach},
  author  = {Mas, Andr{\'e}},
  year    = {2007},
  journal = {Statistical Inference for Stochastic Processes},
  volume  = {10},
  number  = {2},
  pages   = {147--163},
  doi     = {10.1007/s11203-005-0754-3},
}

@book{mathai1992quadratic,
  title     = {Quadratic Forms in Random Variables: Theory and Applications},
  author    = {Mathai, A M. and Provost, Serge B},
  year      = {1992},
  address   = {New York},
  isbn      = {0824786912},
  publisher = {M. Dekker},
  series    = {Statistics, Textbooks and Monographs; v. 126},
}

@article{mbaka2025estimating,
  title   = {Estimating a smooth covariance for functional data},
  author  = {Mbaka, Uche and Ramsay, James O and Carey, Michelle},
  year    = {2025},
  journal = {Computational Statistics \& Data Analysis},
  pages   = {108255},
  doi     = {10.1016/j.csda.2025.108255},
}

@article{mitchell1988bayesian,
  title   = {{Bayesian} variable selection in linear regression},
  author  = {Mitchell, Toby J and Beauchamp, John J},
  year    = {1988},
  journal = {Journal of the American Statistical Association},
  volume  = {83},
  number  = {404},
  pages   = {1023--1032},
  doi     = {10.1080/01621459.1988.10478694},
}

@article{munk2008one,
  title   = {The one- and multi-sample problem for functional data with application to projective shape analysis},
  author  = {Munk, Axel and Paige, Robert and Pang, J and Patrangenaru, Victor and Ruymgaart, Frits},
  year    = {2008},
  journal = {Journal of Multivariate Analysis},
  volume  = {99},
  number  = {5},
  pages   = {815--833},
  doi     = {10.1016/j.jmva.2007.04.007},
}

@article{parzen1961approach,
  title   = {An approach to time series analysis},
  author  = {Parzen, Emanuel},
  year    = {1961},
  journal = {The Annals of Mathematical Statistics},
  volume  = {32},
  number  = {4},
  pages   = {951--989},
  doi     = {10.1214/aoms/1177704840},
}

@book{ramsay2005functional,
  title     = {Functional Data Analysis},
  author    = {Ramsay, James O and Silverman, Bernhard W},
  year      = {2005},
  doi       = {10.1007/b98888},
  edition   = {2nd},
  isbn      = {978-0-387-40080-8},
  publisher = {Springer},
}

@article{richardson1997bayesian,
  title   = {On {Bayesian} analysis of mixtures with an unknown number of components (with discussion)},
  author  = {Richardson, Sylvia and Green, Peter J},
  year    = {1997},
  journal = {Journal of the Royal Statistical Society Series B: Statistical Methodology},
  volume  = {59},
  number  = {4},
  pages   = {731--792},
  doi     = {10.1111/1467-9868.00095},
}

@incollection{robert202250,
  title     = {50 shades of {Bayesian} testing of hypotheses},
  booktitle = {Advancements in {Bayesian} Methods and Implementation},
  author    = {Robert, Christian P},
  year      = {2022},
  volume    = {47},
  pages     = {103--120},
  doi       = {10.1016/bs.host.2022.06.003},
  editor    = {{Srinivasa Rao}, Arni S.R. and Young, G. Alastair and Rao, C.R.},
  publisher = {Elsevier},
  series    = {Handbook of Statistics},
}

@book{robert2007bayesian,
  title     = {The {Bayesian} Choice: From Decision-Theoretic Foundations to Computational Implementation},
  author    = {Robert, Christian P},
  year      = {2007},
  publisher = {Springer},
}

@article{shaffer1995multiple,
  title   = {Multiple hypothesis testing},
  author  = {Shaffer, J. P.},
  year    = {1995},
  journal = {Annual Review of Psychology},
  volume  = {46},
  pages   = {561--584},
  doi     = {10.1146/annurev.ps.46.020195.003021},
}

@article{smaga2020note,
  title   = {A note on repeated measures analysis for functional data},
  author  = {Smaga, {\L}ukasz},
  year    = {2020},
  journal = {AStA Advances in Statistical Analysis},
  volume  = {104},
  number  = {1},
  pages   = {117--139},
  doi     = {10.1007/s10182-018-00348-8},
}

@article{torrecilla2020optimal,
  title   = {Optimal classification of {Gaussian} processes in homo- and heteroscedastic settings},
  author  = {Torrecilla, Jos{\'e} L and Ramos-Carre{\~n}o, Carlos and S{\'a}nchez-Monta{\~n}{\'e}s, Manuel and Su{\'a}rez, Alberto},
  year    = {2020},
  journal = {Statistics and Computing},
  volume  = {30},
  number  = {4},
  pages   = {1091--1111},
  doi     = {10.1007/s11222-020-09937-7},
}

@article{ullah2013applications,
  title   = {Applications of functional data analysis: a systematic review},
  author  = {Ullah, Shahid and Finch, Caroline F},
  year    = {2013},
  journal = {BMC Medical Research Methodology},
  volume  = {13},
  number  = {43},
  pages   = {1--12},
  doi     = {10.1186/1471-2288-13-43},
}

@book{vaart1998asymptotic,
  title     = {Asymptotic Statistics},
  author    = {van der Vaart, Aad W},
  year      = {1998},
  doi       = {10.1017/CBO9780511802256},
  publisher = {Cambridge University Press},
}

@article{wang2016functional,
  title   = {Functional data analysis},
  author  = {Wang, Jane-Ling and Chiou, Jeng-Min and M{\"u}ller, Hans-Georg},
  year    = {2016},
  journal = {Annual Review of Statistics and its Application},
  volume  = {3},
  number  = {1},
  pages   = {257--295},
  doi     = {10.1146/annurev-statistics-041715-033624},
}

@article{wong2019nonparametric,
  title   = {Nonparametric operator-regularized covariance function estimation for functional data},
  author  = {Wong, Raymond K W and Zhang, Xiaoke},
  year    = {2019},
  journal = {Computational Statistics \& Data Analysis},
  volume  = {131},
  pages   = {131--144},
  doi     = {10.1016/j.csda.2018.05.013},
}

@article{xiao2018fast,
  title   = {Fast covariance estimation for sparse functional data},
  author  = {Xiao, Luo and Li, Cai and Checkley, William and Crainiceanu, Ciprian M},
  year    = {2018},
  journal = {Statistics and Computing},
  volume  = {28},
  number  = {3},
  pages   = {511--522},
  doi     = {10.1007/s11222-017-9744-8},
}

@incollection{zellner1986assessing,
  title     = {On assessing prior distributions and {Bayesian} regression analysis with $g$-prior distributions},
  booktitle = {Bayesian Inference and Decision Techniques: Essays in Honor of Bruno de Finetti},
  author    = {Zellner, Arnold},
  year      = {1986},
  pages     = {233--243},
  editor    = {Goel, P. and Zellner, A.},
  publisher = {Elsevier},
}

@article{zhang2010two,
  title   = {Two samples tests for functional data},
  author  = {Zhang, Chongqi and Peng, Heng and Zhang, Jin-Ting},
  year    = {2010},
  journal = {Communications in Statistics - Theory and Methods},
  volume  = {39},
  number  = {4},
  pages   = {559--578},
  doi     = {10.1080/03610920902755839},
}

@book{zhang2013analysis,
  title     = {Analysis of Variance for Functional Data},
  author    = {Zhang, Jin-Ting},
  year      = {2013},
  doi       = {10.1201/b15005},
  publisher = {Chapman and Hall/CRC},
}

\end{document}